\documentclass[journal]{IEEEtran}
\ifCLASSINFOpdf
\else
\fi

\newcommand{\renyi}{R\'{e}nyi }

\newcommand{\cTheta}{\widetilde{\Theta}}
\newcommand{\ctheta}{\tilde{\theta}}
\newcommand{\dtheta}{\bar{\theta}}
\newcommand{\ctpen}{\tilde{L}}
\newcommand{\copen}{L}
\newcommand{\esigma}{\Sigma^{\lambda}_{\theta}}

\newcommand{\rs}[1]{\mathrsfs{#1}}

\newcommand{\thatis}{{i.e.}}
\newcommand{\eqnum}[1]{(\ref{#1})}
\newcommand{\eqnumn}[1]{(\ref{#1})}
\newcommand{\trace}{{\rm Tr}}
\newcommand{\mib}[1]{\mbox{\boldmath$#1$}}
\newcommand{\typicalset}{A_{\epsilon}^{n}}
\newcommand{\inprob}{P_{\epsilon}^{n}}
\newcommand{\condexp}{E^{n}_{\epsilon}}

\newcommand{\twop}{\tilde{p}_{2}}
\newcommand{\twol}{\tilde{L}_{2}}
\newcommand{\twopc}{p_{2}}

\usepackage{ascmac}
\usepackage{amsmath}
\usepackage{amsthm}
\usepackage{amsfonts}
\usepackage{calrsfs}
\usepackage[dvips]{graphicx}
\usepackage{float}
\usepackage{color}

\newtheorem{theorem}{\indent Theorem}

\newtheorem{condition}{\indent Condition}
\newtheorem{corollary}{\indent Corollary}
\newtheorem{definition}[theorem]{\indent Definition}

\newtheorem{lemma}{\indent Lemma}

\begin{document}
\title{Minimum Description Length Principle in Supervised Learning with Application to Lasso}
\author{Masanori~Kawakita~
        and~Jun'ichi~Takeuchi
\thanks{This work was supported in part by JSPS KAKENHI Grant Number
25870503 and the Okawa Foundation for Information and
Telecommunications. This material will be presented in part at the 33rd
International Conference on Machine Learning in New York city, NY, USA.

M. Kawakita and J. Takeuchi are with the faculty of Information Science and
        Electrical Engineering, Kyushu University, 744, Motooka,
        Nishi-Ku, Fukuoka city, 819-0395 Japan 
        (e-mail: kawakita@inf.kyushu-u.ac.jp).}}

\maketitle

 \begin{abstract}
The minimum description length (MDL) principle in supervised
  learning is studied. One of the most important theories for the MDL
  principle is Barron and Cover's theory (BC theory), which gives a
  mathematical justification of the MDL principle. 
 The original BC theory, however, can be applied to supervised learning only
  approximately and limitedly.
 Though Barron et al. recently succeeded in removing a similar approximation in case of unsupervised learning, their
 idea cannot be essentially applied to supervised learning in
 general. To overcome this issue, an extension of BC theory to
  supervised learning is proposed. 
 The derived risk bound has several advantages inherited from the
 original BC theory. First, the risk bound holds for finite sample size. Second,
 it requires remarkably few assumptions. Third, the risk bound has a
 form of redundancy of the two-stage code for the MDL procedure. Hence, the proposed extension gives a mathematical justification of the
  MDL principle to supervised learning like the original BC theory. 
As an important example of application, new risk and (probabilistic)
  regret bounds of lasso with random design are derived. The derived risk bound holds for any
 finite sample size $n$ and feature number $p$ even 
 if $n\ll p$ without boundedness of features in contrast to the past
 work. Behavior of the regret bound is investigated by numerical
 simulations. We believe that this is the first extension of BC theory
 to general supervised learning with random design without
 approximation. 
 \end{abstract}


\section{Introduction}\label{introduction}
There have been various techniques to evaluate performance of
machine learning methods theoretically. Taking lasso \cite{tibshirani96a} as an
example, lasso has been analyzed by nonparametric
statistics \cite{Bunea2007,Bunea2007a,zhang09,Bickel2009b}, empirical
process \cite{bartlettetal12}, statistical physics
\cite{bayatietal10,bayatimontanari12,bayatietal13} and so on.
In general, most of these techniques require either asymptotic assumption (sample number $n$ and/or feature
number $p$ go to infinity) or
various technical assumptions like boundedness of features or moment
conditions. Some of them are much restrictive for practical use. In
this paper, we try to develop another way for performance evaluation
of machine learning methods with as few assumptions as
possible. An important candidate for this purpose is Barron and Cover's
theory (BC theory), which is one of the most famous results 
for the minimum description length (MDL) principle.  The MDL principle
\cite{Rissanen1978,Barron1998,grunwald07,grunwaldetal05,takeuchi14} claims that the
shortest description of a given set of data leads to the best hypotheses
about the data source. A famous model selection criterion based
on the MDL principle was proposed by Rissanen \cite{Rissanen1978}. This criterion
corresponds to a codelength of a two-stage code in which one encodes a
statistical model to encode data and then the
data are encoded with the model. In this case, an MDL estimator is
defined as the minimizer of the total codelength of this two-stage code.
BC theory \cite{Barron1991} guarantees that a risk of the MDL estimator
in terms of the \renyi divergence \cite{renyi61} is tightly bounded from
above by redundancy of the corresponding two-stage code. Because this
result means that the shortest description of the data by the
two-stage code yields the smallest risk upper bound, this result gives a
mathematical justification of the MDL principle.  Furthermore, BC theory
holds for finite $n$ without any complicated technical
conditions. However, BC theory has been applied to supervised learning
only approximately or limitedly. The original BC theory seems to be
widely recognized that it can be applicable to both unsupervised and
supervised learning. Though it is not false, BC theory actually cannot
be applied to supervised learning without a certain condition (Condition
\ref{criticalcondition} defined in Section \ref{bctheory}).  This condition is critical in a
sense that lack of this condition breaks a key technique of BC
theory. The literature \cite{Yamanishi1992} is the only example of
application of BC theory to supervised learning to our knowledge. His work assumed a specific
setting, where Condition \ref{criticalcondition} can be satisfied. However, the risk
bound may not be sufficiently tight due to imposing Condition
\ref{criticalcondition} forcedly, which will be explained in Section \ref{bctheory}.  Another
well-recognized disadvantage is the necessity of quantization of
parameter space. Barron et al. 
proposed a way to avoid the quantization and derived a risk bound of
lasso \cite{Barron2008,chatterjeebarron14} as an example. However, their idea cannot be applied to supervised learning in
general.  The main difficulty stems from Condition
\ref{criticalcondition} as explained later. It is thus
essentially difficult to solve. Actually, their risk bound of lasso was
derived with fixed design only (\thatis, essentially unsupervised
setting). The fixed design, however, is not satisfactory to evaluate
generalization error of supervised learning.  In this paper, we propose
an extension of BC theory to supervised learning without quantization in
random design cases. The derived risk bound inherits most of advantages
of the original BC theory. The main term of the risk bound has again a form of
redundancy of two-stage code. Thus, our extension also gives a 
mathematical justification of the MDL principle in supervised
learning. It should be remarked that, however, an additional condition
is required for an exact redundancy interpretation. 
We also derive new risk and regret bounds of lasso with random design as its application
under normality of features. This application is not trivial at all and
requires much more effort than both the above extension itself and the
derivation in fixed design cases.
We will try to derive those bounds in a manner not specific to our setting
but rather applicable to several other settings.  
Interestingly, the redundancy and
regret interpretation for the above bounds are exactly justified without
any additional condition in the
case of lasso.  
The most advantage of our theory is that it requires almost no
assumptions: neither asymptotic assumption ($n<p$ is also allowed),
bounded assumptions, moment conditions nor other technical conditions.
Especially, it is remarkable that our risk evaluation holds for finite
$n$ without necessity of boundedness of features though the employed
loss function (the \renyi divergence) is not bounded. Behavior of the
regret bound will be investigated by numerical simulations.
It may be worth noting that, despite we tried several other
approaches in order to extend BC theory to supervised learning, we can
hardly derive a risk bound of lasso as tight as meaningful by using
them. We believe that our proposal is currently the unique choice that could give a meaningful risk bound. 

This paper is organized as follows. Section \ref{setting}
introduces an MDL estimator in supervised learning.
We briefly review BC theory and its recent progress in Section \ref{bctheory}.
The extension of BC theory to supervised learning will appear in Section
\ref{extension}. We derive new risk and regret bounds of lasso in Section
\ref{lasso}.
All proofs of our results are given in Section \ref{proofs}. Section \ref{simulation} contains numerical
simulations. A conclusion will appear in Section \ref{conclusion}.
\section{MDL Estimator in Supervised Learning}\label{setting}
Suppose that we have $n$ training data
$(x^n,y^n):=\{(x_i,y_i)\in \rs{X}\times \rs{Y}|i=1,2,\cdots, n\}$ generated
from $\bar{p}_*(x^n,y^n)=q_*(x^n)p_*(y^n|x^n)$, where $\rs{X}$ is a domain
of feature vector $x$ and $\rs{Y}$ could be $\Re$ (regression) or a finite set
(classification) according to target problems. 
Here, the sequence $(x_1,y_1),(x_2,y_2),\cdots$ is not necessarily independently and
   identically distributed (i.i.d.) but can be a stochastic process in general.
We write the $j$th component of the $i$th sample as $x_{ij}$.
To define an MDL estimator according to the notion of
two-stage code \cite{Rissanen1978}, we need to describe data itself and
a statistical model used to describe the data too.
Letting $\tilde{L}(x^n,y^n)$ be the codelength of the two-stage code to
describe $(x^n,y^n)$, $\tilde{L}(x^n,y^n)$ can be decomposed as
\[
 \tilde{L}(x^n,y^n)=\tilde{L}(x^n)+\tilde{L}(y^n|x^n)
\]
by the chain rule. Since a goal of supervised learning is to estimate $p_*(y^n|x^n)$, we
need not estimate $q_*(x^n)$. In view of the MDL principle, this implies
that $\tilde{L}(x^n)$ (the description length of
$x^n$) can be ignored. Therefore, we only consider the encoding of
$y^n$ given $x^n$ hereafter. 
This corresponds to a description scheme in which an encoder and a
decoder share the data $x^n$.
To describe $y^n$ given $x^n$, we use a parametric model
$p_{\theta}(y^n|x^n)$ with parameter $\theta\in \Theta$.
The parameter space $\Theta$ is a certain continuous space or a union of
continuous spaces.
Note that, however, the continuous parameter cannot be encoded. Thus, we
need to quantize the parameter space $\Theta$ as $\cTheta(x^n)$.
According to the notion of the two-stage code, we need to describe not
only $y^n$ but also the model used to describe $y^n$ (or equivalently the parameter
$\ctheta\in \cTheta(x^n)$) given $x^n$. Again by the chain rule, such a codelength can be decomposed as
\[
 \tilde{L}(y^n,\ctheta|x^n)=\tilde{L}(y^n|x^n,\ctheta)+\tilde{L}(\ctheta|x^n).
\]
Here, $\tilde{L}(y^n|x^n,\ctheta)$ expresses a codelength to
describe $y^n$ using $p_{\ctheta}(y^n|x^n)$, which is, needless to say, $-\log
p_{\ctheta}(y^n|x^n)$. On the other hand, $\ctpen(\ctheta|x^n)$ expresses a
codelength to describe the model $p_{\ctheta}(y^n|x^n)$ itself. 
Note that
$\ctpen(\tilde{\theta}|x^n)$ must satisfy Kraft's inequality
\begin{equation}
\sum_{\ctheta\in \cTheta(x^n)}\exp(-\ctpen(\ctheta|x^n))\le 1.  \nonumber\label{kraft}
\end{equation}
The MDL estimator is defined by the minimizer of
the above codelength:
\[
 \ddot{\theta}(x^n,y^n):=\arg \min_{\ctheta\in \cTheta(x^n)}\big\{-\log
 p_{\ctheta}(y^n|x^n)+\ctpen(\ctheta|x^n)\big\}.
\]
Let us write the minimum description length attained by the two-stage
 code as
 \[
 \twol (y^n|x^n):=-\log  p_{\ddot{\theta}}(y^n|x^n)+\ctpen(\ddot{\theta}|x^n).
 \]
Because $\twol$ also satisfies Kraft's inequality with respect
 to $y^n$ for each $x^n$, it is interpreted as a codelength of a prefix two-stage code. Therefore,
\[
 \twop (y^n|x^n):=\exp(-\twol(y^n|x^n)) 
\]
 is a conditional
sub-probability distribution corresponding to the two-stage code.

\section{Barron and Cover's Theory}\label{bctheory}
We briefly review Barron and Cover's theory (BC theory) and its recent
progress in view of supervised learning though they discussed basically
unsupervised learning (or supervised learning with fixed design).
In BC theory, the \renyi divergence \cite{renyi61} between $p(y|x)$ and $r(y|x)$ with
order $\lambda\in (0,1)$
\begin{equation}
 d^n_\lambda(p,r)=-\frac{1}{1-\lambda}
\log E_{q_*(x^n)p(y^n|x^n)} \left(\frac{r(y^n|x^n)}{p(y^n|x^n)}\right)^{1-\lambda}\label{renyidiv}
\end{equation}
is used as a loss function. The \renyi divergence converges to
Kullback-Leibler (KL) divergence
\begin{equation}
 \rs{D}^n(p,r):=\int q_*(x^n)p(y^n|x^n)\left(\log
					\frac{p(y^n|x^n)}{r(y^n|x^n)}\right)dx^ndy^n \label{kldiv}
\end{equation}
as $\lambda\rightarrow 1$, \thatis,
\begin{equation}
 \lim_{\lambda\rightarrow 1}d^n_{\lambda}(p,r)=\rs{D}^n(p,r) \label{renyiKL}
\end{equation}
for any $p,r$ . We also note that the \renyi divergence at
$\lambda=0.5$ is equal to Bhattacharyya divergence \cite{bhattacharyya43}
\begin{equation}
 d^n_{0.5}(p,r)=-2\log \int q_*(x^n)\sqrt{p(y^n|x^n)r(y^n|x^n)}dx^ndy^n.\label{bhattacharyya}
\end{equation}
We drop $n$ of each divergence like $d_{\lambda}(p,r)$ if it
is defined with a single random variable, \thatis,
\[
 d_{\lambda}(p,r)=-\frac{1}{1-\lambda}\log E_{q_*(x)p(y|x)}\left(\frac{r(y|x)}{p(y|x)}\right)^{1-\lambda}.
\]
BC theory requires the model description length to satisfy a little bit
stronger Kraft's inequality defined as follows.
\begin{definition}
 Let $\beta$ be a real number in $(0,1)$. We say that a
 function $h(\ctheta)$ satisfies $\beta$-stronger Kraft's inequality
 if
 \[
  \sum_{\ctheta}\exp(-\beta h(\ctheta))\le 1,
 \]
 where the summation is taken over a range of $\ctheta$ in its context.
\end{definition}
The following condition
is indispensable for application of BC theory to supervised learning. 
\begin{condition}[indispensable condition]\label{criticalcondition}
Both the quantized space and the model description length are independent of
$x^n$, \thatis, 
\begin{equation}
\cTheta(x^n)=\cTheta,\quad \ctpen(\ctheta|x^n)=\ctpen(\ctheta).\label{indcondition}
\end{equation}
\end{condition}
Under Condition \ref{criticalcondition}, BC
theory \cite{Barron1991} gives the following two theorems for supervised
learning.
  Though these theorems were shown only for the case of
  Hellinger distance in the original literature \cite{Barron1991}, we
  state these theorems with the \renyi divergence.  
  \begin{theorem}\label{bctheorem1}
   Let $\beta$ be a real number in $(0,1)$.
Assume that $\ctpen$ satisfies $\beta$-stronger Kraft's inequality. 
Under Condition \ref{criticalcondition}, 
\begin{eqnarray}
\lefteqn{E_{\bar{p}_*(x^n,y^n)}
d^n_\lambda(p_*,p_{\ddot{\theta}})}\nonumber\\
 &\leq &
E_{\bar{p}_*(x^n,y^n)}\left[\inf_{\ctheta\in
		       \Theta}\left\{\log\frac{p_*(y^n|x^n)}{p_{\ctheta}(y^n|x^n)}+\ctpen(\ctheta)\right\}\right]\label{origbound}
\\
 &=&
E_{\bar{p}_*(x^n,y^n)}\log\frac{p_*(y^n|x^n)}{\twop(y^n|x^n)}\label{origbound2}
\end{eqnarray}
for any $\lambda\in (0,1-\beta]$.   
 \end{theorem}
  \begin{theorem}\label{bctheorem2}
   Let $\beta$ be a real number in $(0,1)$. 
   Assume that $\ctpen$ satisfies $\beta$-stronger Kraft's inequality.
Under Condition \ref{criticalcondition},  
\begin{eqnarray*}
\Pr\Bigl(
 \frac{d^n_\lambda(p_*,p_{\ddot{\theta}})}{n}
 -
\frac{1}{n}
\log \frac{p_{*}(y^n|x^n)}{\twop(y^n|x^n)}
\geq \tau \Bigr)
 \leq 
e^{-n\tau \beta }
\end{eqnarray*}
   for any $\lambda\in (0,1-\beta]$.
  \end{theorem}
  Since the right side of (\ref{origbound2}) is just the redundancy of
  the prefix two-stage code, Theorem \ref{bctheorem1} implies that we
  obtain the smallest upper bound of the risk by compressing the data
  most with the two-stage code. That is, Theorem \ref{bctheorem1} is
  a mathematical justification of the MDL principle. 
  We remark that, by interchanging the infimum and the expectation of
  (\ref{origbound}), the right side of (\ref{origbound}) becomes a quantity
  called ``index of resolvability'' \cite{Barron1991}, which is an upper
  bound of redundancy. 
  It is remarkable that BC theory requires no assumption except
  Condition \ref{criticalcondition} and
  $\beta$-stronger Kraft's inequality. However, Condition \ref{criticalcondition} is a
  somewhat severe restriction.
Both the quantization and the model description length can depend on
  $x^n$ in the definitions. In view of the MDL principle, this is
  favorable because the total description length can be minimized
  according to $x^n$ flexibly. If we use the model description length that is uniform over
  $\rs{X}^n$ in contrast, the total codelength must be longer in
  general. Hence, data-dependent model description length is more
  desirable.
Actually, this observation suggests that the bound derived in \cite{Yamanishi1992} may not be sufficiently tight.
In addition, the restriction by Condition \ref{criticalcondition} excludes a
practically important case `lasso with column normalization' (explained
  below) from the scope of application.  
However, it is essentially difficult to remove this restriction as
noted in Section \ref{introduction}. Another concern is quantization. The quantization for the encoding is natural in view of
the MDL principle.
Our target, however, is an application to usual estimators or machine
  learning algorithms themselves including lasso. A trivial example of such an application is a penalized maximum
likelihood estimator (PMLE)  
\begin{eqnarray*}
  \hat{\theta}(x^n,y^n)\!\!\!&:=&\!\!\!\arg \min_{\theta\in \Theta}\big\{-\log
   p_{\theta}(y^n|x^n)+ \copen(\theta|x^n)\big\},
\end{eqnarray*}
where $\copen:\Theta\times \rs{X}^n\rightarrow [0,\infty)$ is a certain
penalty.
Similarly to the quantized case, let us define
\begin{eqnarray*}
\twopc(y^n|x^n):=p_{\hat{\theta}}(y^n|x^n)\cdot
  \exp(-\copen(\hat{\theta}|x^n)), 
\end{eqnarray*}
that is, 
\[
  -\log \twopc(y^n|x^n)= \min_{\theta\in
  \Theta}\left\{-\log p_{\theta}(y^n|x^n)+L(\theta|x^n)\right\}. 
\]
Note that, however, $p_2(y^n|x^n)$ is not necessarily a
sub-probability distribution in contrast to the quantized case, which
will be discussed in detail in Section \ref{extension}. 
PMLE is a wide class of estimators including many
useful methods like Ridge regression \cite{hastieetal01}, lasso, Dantzig
Selector \cite{candestao07} and any Maximum-A-Posteriori estimators of Bayes estimation. If we can accept $\ddot{\theta}$ as an
approximation of $\hat{\theta}$ (by taking $\ctpen=\copen$), we have a risk
bound by direct application of BC theory. 
However, the quantization is unnatural in view of
machine learning application. Besides, we cannot use any data-dependent
$\copen$. Barron et al. proposed an important
notion `{\it risk validity}' to remove the quantization
\cite{barronetal08b,chatterjeebarron14,Chatterjee2014}.  
 \begin{definition}[risk validity]\label{riskvalidity}
Let $\beta$ be a real number in $(0,1)$ and $\lambda$ be a real number
  in $(0,1-\beta]$. For fixed $x^n$, we say
  that a penalty function $\copen(\theta|x^n)$ is risk valid if there exist a quantized space
  $\cTheta(x^n) \subset \Theta$ and a model description length $\ctpen(\ctheta|x^n)$ satisfying $\beta$-stronger Kraft's
  inequality such that $\cTheta(x^n)$ and $\ctpen(\ctheta|x^n)$ satisfy 
 \begin{eqnarray}
&&\hskip-7mm\forall y^n\in \rs{Y}^n,\,  \max_{\theta\in \Theta}
 \Bigl\{\!d^n_{\lambda}
 (p_*,p_{\theta}|x^n)\!-\!\log\frac{p_*(y^n|x^n)}{p_{\theta}(y^n|x^n)}\!-\!
 \copen(\theta|x^n)\Bigr\}\nonumber\\
  &&  \hskip-7mm
\le \max_{\ctheta \in \cTheta(x^n)}\!
\Bigl\{\!
d^n_{\lambda}(p_*,p_{\ctheta}|x^n)
\!-\!
\log \frac{p_*(y^n|x^n)}{p_{\ctheta}(y^n|x^n)}
\!-\! \ctpen(\ctheta|x^n)\!
\Bigr\}, \label{rv}
 \end{eqnarray}
  where
\[
 d^n_{\lambda}(p,r|x^n):= -\frac{1}{1-\lambda}
\log E_{p(y^n|x^n)} \Bigl(\frac{r(y^n|x^n)}{p(y^n|x^n)}\Bigr)^{1-\lambda}.
\]
 \end{definition}
 Note that their original definition in \cite{chatterjeebarron14} was
 presented only for the case where $\lambda=1-\beta$. Here,
 $d(p,r|x^n)$ is the \renyi divergence for fixed design ($x^n$ is
 fixed). Hence, $d^n_{\lambda}(p,r|x^n)$ does not depend on $q_*(x^n)$ in
 contrast to the \renyi divergence for random design
 $d^n_{\lambda}(p,r)$ defined by (\ref{renyidiv}). 
  Barron et al. proved that $\hat{\theta}$ has bounds similar to Theorems
\ref{bctheorem1} and \ref{bctheorem2} for any risk valid
penalty in the fixed design case. Their way is excellent because it does not require any
additional condition other than the risk validity. However, the risk
 evaluation only for a particular $x^n$ like
 $E_{p_*(y^n|x^n)}[d^n_{\lambda}(p_*,p_{\hat{\theta}}|x^n)]$ is unsatisfactory for supervised
 learning. In 
order to evaluate the so-called `generalization error' of supervised
learning, we need to evaluate the risk with random design, \thatis,
 $E_{\bar{p}_*(x^n,y^n)}[d^n_{\lambda}(p_*,p_{\hat{\theta}})]$. However, it is
essentially difficult to apply their idea to random design cases as it
is. Let us explain this by using lasso as an example.
The readers unfamiliar to lasso can refer to the head of
 Section \ref{lasso} for its definition. 
By extending the definition of risk validity to random design
straightforwardly, we obtain the following definition.
 \begin{definition}[risk validity in random design]\label{riskvalidityinrd}
Let $\beta$ be a real number in $(0,1)$ and $\lambda$ be a real number
  in $(0,1-\beta]$. We say that a penalty function
  $\copen(\theta|x^n)$ is risk valid if there exist a quantized space $\cTheta \subset \Theta$ and
  a model description length $\ctpen(\ctheta)$ satisfying $\beta$-stronger Kraft's inequality
   such that  
 \begin{eqnarray}
  &&\hskip-7mm\forall x^n\in \rs{X}^n,\ y^n\in \rs{Y}^n,\nonumber\\
  &&\max_{\theta\in \Theta}
 \Bigl\{d^n_{\lambda}
 (p_*,p_{\theta})\!-\!\log\frac{p_*(y^n|x^n)}{p_{\theta}(y^n|x^n)}\!-\!
 \copen(\theta|x^n)\Bigr\}\nonumber\\
  &&  \hskip-7mm
\le \max_{\ctheta \in \cTheta}
\Bigl\{
d^n_{\lambda}(p_*,p_{\ctheta})
-
\log \frac{p_*(y^n|x^n)}{p_{\ctheta}(y^n|x^n)}
- \ctpen(\ctheta)
\Bigr\}. \label{rvrd}
 \end{eqnarray}
 \end{definition}
In contrast to the fixed design case, (\ref{rv}) must hold not only for a fixed
$x^n\in \rs{X}^n$ but also for all $x^n\in \rs{X}^{n}$.
In addition, $\cTheta$ and $\ctpen(\ctheta)$ must be independent of
$x^n$ due to Condition \ref{criticalcondition}. 
The
form of \renyi divergence $d^n_{\lambda}(p_*,p_{\theta})$ also differs from
$d^n_{\lambda}(p_*,p_{\theta}|x^n)$ of the fixed design case in general.
Let us rewrite (\ref{rvrd}) equivalently as
\begin{eqnarray}
\lefteqn{\forall x^n \in \rs{X}^n,\ \forall y^n\in \rs{Y}^n,\ \forall
\theta\in \Theta,} \nonumber\\  
&&\hskip-7mm\min_{\ctheta \in \cTheta}
\Bigl\{
d^n_{\lambda}(p_*,p_{\theta})-
d^n_{\lambda}(p_*,p_{\ctheta})
+\log
\frac{p_{\theta}(y^n|x^n)}{p_{\ctheta}(y^n|x^n)}+\ctpen(\ctheta)\Bigr\}\nonumber \\
&& \hskip-7mm\le \copen(\theta|x^n).\label{rv2}
\end{eqnarray}
For short, we write the inside part of the
minimum of the left side of (\ref{rv2}) as
$H(\theta,\tilde{\theta},x^n,y^n)$. We need to evaluate $\min_{\tilde{\theta}}\{H(\theta,\tilde{\theta},x^n,y^n)\}$ in order to
derive risk valid penalties.  However,
it seems to be considerably
difficult. To our knowledge, the technique used by Chatterjee and Barron
\cite{chatterjeebarron14} is the best way to evaluate it, so that we
also employ it in this paper. A key premise of their idea is that taking
$\tilde{\theta}$ close to $\theta$ is not a bad choice to evaluate $\min_{\ctheta}H(\theta,\ctheta,x^n,y^n)$.  Regardless of whether it is true or not, this premise seems to be natural
and meaningful in the following sense. If we quantize the parameter
space finely enough, the quantized estimator $\ddot{\theta}$ is expected to behave almost
similarly to $\hat{\theta}$ with the same penalty and is expected to have a similar risk
bound.  If we take $\tilde{\theta}=\theta$, then
$H(\theta,\tilde{\theta},x^n,y^n)$ is equal to $\ctpen(\theta)$, which
implies that $\ctpen(\theta)$ is a risk valid penalty and has a risk
bound similar to the quantized case. Note that, however, we cannot match
$\tilde{\theta}$ to $\theta$ exactly because $\tilde{\theta}$ must be on
the fixed quantized space $\cTheta$. So, Chatterjee and Barron randomized
$\tilde{\theta}$ on the grid points on $\cTheta$ around $\theta$ and evaluate the
expectation with respect to it. This is clearly justified because
$\min_{\tilde{\theta}}\{H(\theta,\tilde{\theta},x^n,y^n)\}\le
E_{\tilde{\theta}}[H(\theta,\tilde{\theta},x^n,y^n)]$.
By using a carefully tuned randomization, they succeeded in removing the dependency of
$E_{\tilde{\theta}}[H(\theta,\tilde{\theta},x^n,y^n)]$ on $y^n$. Let us write the resultant expectation
as $H'(\theta,x^n):=E_{\tilde{\theta}}[H(\theta,\tilde{\theta},x^n,y^n)]$
for convenience. Any upper bound $\copen(\theta|x^n)$ of $H'(\theta,x^n)$ is a risk
valid penalty. By this fact, risk valid penalties should basically depend on
$x^n$ in general. If not ($\copen(\theta|x^n)=\copen(\theta)$),
$\copen(\theta)$ must bound $\max_{x^n}H'(\theta,x^n)$, which makes
$L(\theta)$ much larger. This is again unfavorable in view of the MDL
principle. In particular, $H'(\theta,x^n)$ includes an unbounded term in
linear regression cases with regard to $x^n$, which originates from the third term of the left side
of (\ref{rv2}). This can be seen by checking Section III of \cite{chatterjeebarron14}. Though their setting is fixed design,
this fact is also true for the random design. Hence, as long
as we use their technique, derived risk valid penalties must depend on
$x^n$ in linear regression cases. However, the $\ell_1$ norm used in the
usual lasso does not depend on $x^n$.  Hence, the risk
validity seems to be useless for lasso.  However, the following 
weighted $\ell_1$ norm
\begin{eqnarray*}
 \|\theta\|_{w,1}&:=&\sum_{j=1}^pw_j|\theta_j|,\\
 \mbox{where}&& w:=(w_1,\cdots, w_p)^T,\quad w_j:=\sqrt{\frac{1}{n}\sum_{i=1}^nx_{ij}^2}
\end{eqnarray*}
plays an important role here. 
The lasso with this weighted $\ell_1$ norm
is equivalent to an ordinary lasso with column normalization such that
each column of the design matrix has the same norm. The column normalization is theoretically
and practically important. Hence, we try to find a risk valid penalty
of the form $L_1(\theta|x^n)=\mu_1\|\theta\|_{w,1}+\mu_2$, where $\mu_1$ and $\mu_2$ are real coefficients. Indeed, there seems to be no other useful
penalty dependent on $x^n$ for the usual lasso. In contrast to fixed design cases,
however, there are severe difficulties to derive a meaningful risk bound
with this penalty. We explain this intuitively.
The main difficulty is caused by Condition \ref{criticalcondition}. As described
above, our strategy is to take $\tilde{\theta}$ close to
$\theta$. Suppose now that it is ideally almost realizable for any choice
of $x^n,y^n,\theta$. This implies that
$H(\theta,\tilde{\theta},x^n,y^n)$ is almost equal to
$\tilde{L}(\theta)$. On the other hand, for each fixed $\theta$, the
weighted $\ell_1$ norm of $\theta$ can be arbitrarily small by making
$x^n$ small accordingly. Therefore, the penalty
$\mu_1\|\theta\|_{w,1}+\mu_2$ is almost equal to $\mu_2$ in this case. This implies that $\mu_2$ must bound
$\max_{\theta}\ctpen(\theta)$, which is infinity in general.
If $\ctpen$ depended on $x^n$, we could resolve this problem.
However, $\ctpen$ must be independent of $x^n$. 
This issue does not seem to be specific to lasso.
Another major issue is the \renyi divergence
$d^n_{\lambda}(p_*,p_{\theta})$.
In the fixed design case, the
\renyi divergence $d^n_{\lambda}(p_*,p_{\theta}|x^n)$ is a simple convex
function in terms of $\theta$,
which makes its analysis easy. In contrast, the \renyi divergence
$d^n_{\lambda}(p_*,p_{\theta})$ in case of random design is not convex
and more complicated than that of fixed design cases, which makes it 
difficult to analyze.
We will describe why the non-convexity of loss function makes the
analysis difficult in Section \ref{someremarks}. 
The difficulties that we face when we use the techniques of
\cite{chatterjeebarron14} in the random design case are not limited to them. 
We do not explain them here because it requires the readers to understand
their techniques in detail. However, we only remark that these difficulties seem
to make their techniques useless for supervised learning with random design.
We propose a remedy to solve these issues in a lump in the next section.
\section{Main Results}
In this section, we propose a way to extend BC theory to supervised
learning and derive a new risk bound of lasso.
\subsection{Extension of BC Theory to Supervised
  Learning}\label{extension}
There are several possible approaches to extend BC theory to supervised
learning. A major concern is how tight a resultant risk bound is.
Below, we propose a way that gives a tight risk upper bound for
at least lasso. A key idea is to modify the risk validity
condition by introducing a so-called typical set of $x^n$.
We postulate that a probability distribution of stochastic process $x_1,x_2,\cdots,$ is a member of
a certain class $\mathcal{P}_x$.
Furthermore, we define $\mathcal{P}^n_x$ by the set of marginal
distribution of $x_1,x_2,\cdots, x^n$ of all elements of
$\mathcal{P}_x$.
We assume that we can define a typical set $\typicalset$ for each $q_*\in \mathcal{P}^n_x$, \thatis, $\mbox{Pr}(x^n\in \typicalset) \rightarrow 1$ as $n\rightarrow
\infty$. This is possible if $q_*$ is stationary and ergodic for
example. See \cite{coverthomas06} for detail. For short, $\mbox{Pr}(x^n\in \typicalset)$ is written as
$\inprob $ hereafter. 
We modify the risk validity by using the typical set. 
 \begin{definition}[$\epsilon$-risk validity]\label{epsilonriskvalid}
Let $\beta,\epsilon$ be real numbers in $(0,1)$ and $\lambda$ be a real number
  in $(0,1-\beta]$. We say that $\copen(\theta|x^n)$ is $\epsilon$-risk valid for $(\lambda,\beta,\mathcal{P}^n_x,
 \typicalset)$ if for any $q_*\in
 \mathcal{P}^n_x$, there exist a quantized subset $\cTheta(q_*) \subset \Theta$ and a model
 description length $\ctpen(\ctheta|q_*)$ satisfying $\beta$-stronger
  Kraft's inequality such that 
\begin{align*} 
\forall & x^n \in \typicalset,\  \forall y^n\in \rs{Y}^n, \\
&\max_{\theta \in \Theta}
\Bigl\{
d^n_{\lambda}(p_*,p_{\theta})
-
\log \frac{p_*(y^n|x^n)}{p_{\theta}(y^n|x^n)}
- \copen(\theta|x^n)
 \Bigr\}\\
 \le  &
\max_{\ctheta \in \cTheta(q_*)}
\Bigl\{
d^n_{\lambda}(p_*,p_{\ctheta})
-
\log \frac{p_*(y^n|x^n)}{p_{\ctheta}(y^n|x^n)}
- \ctpen(\ctheta|q_*)
\Bigr\}.
\end{align*}
 \end{definition}
Note that both $\cTheta$ and $\ctpen$ can depend on the unknown
distribution $q_*(x^n)$. This is not
problematic because the final penalty $\copen$ does not depend on
the unknown $q_*(x^n)$. A difference from \eqnum{rv2} is the restriction of
the range of $x^n$ onto the typical set.
From here to the next section, we will see how this small change solves
the problems described in the previous section. First, we show what can
be proved for $\epsilon$-risk valid penalties. 
\begin{theorem}[risk bound]\label{mytheorem1}
Define $\condexp$ as a conditional expectation with regard to
 $\bar{p}_*(x^n,y^n)$ given that $x^n\in \typicalset$.  
Let $\beta,\epsilon$ be arbitrary real numbers in $(0,1)$. 
For any $\lambda\in (0,1-\beta]$, if $\copen(\theta|x^n)$ is
 $\epsilon$-risk valid for $(\lambda,\beta,\mathcal{P}^n_x,\typicalset)$, 
\begin{eqnarray}
\condexp d^n_\lambda(p_*,p_{\hat{\theta}})
\le
\condexp
\log \frac{p_*(y^n|x^n)}{\twopc(y^n|x^n)}
+ \frac{1}{\beta} \log \frac{1}{\inprob }.\label{finriskbound}
\end{eqnarray}
\end{theorem}
\begin{theorem}[regret bound]\label{mytheorem2}
 Let $\beta,\epsilon$ be arbitrary real numbers in $(0,1)$.
 For any $\lambda\in (0,1-\beta]$, if $\copen(\theta|x^n)$ is
 $\epsilon$-risk valid for
 $(\lambda,\beta,\mathcal{P}^n_x,\typicalset)$,
  \begin{align} 
&\Pr
\Bigl(
\frac{d^n_\lambda(p_*,p_{\hat{\theta}})}{n}
-
\frac{1}{n}
\log \frac{p_*(y^n|x^n)}{\twopc(y^n|x^n)}
 \ge \tau
\Bigr)
\nonumber\\
\le & \,\,\exp(- n\tau \beta )+1-\inprob .\label{regretbound}
  \end{align}
\end{theorem}
A proof of Theorem \ref{mytheorem1} is described in Section \ref{mytheorem1proof},
while a proof of Theorem \ref{mytheorem2} is described in Section
\ref{mytheorem2proof}.
Note that both bounds become tightest when $\lambda=1-\beta$ because
the \renyi divergence $d^n_{\lambda}(p,r)$ is monotonically increasing in
terms of $\lambda$ (see \cite{grunwald07} for example).
We call the quantity $-\log (1/p_2(y^n|x^n))-(-\log
(1/p_*(y^n|x^n)))$ in Theorem \ref{mytheorem2} `regret' of the two-stage
code $p_2$ on the given data
$(x^n,y^n)$ in this paper, though the ordinary regret is defined as the
codelength difference from $\log(1/p_{\hat{\theta}_{mle}}(y^n|x^n))$,
where $\hat{\theta}_{mle}$ denotes the maximum likelihood estimator.  
Compared to the usual BC
theory, there is an additional term $(1/\beta)\log (1/\inprob)$ in the
risk bound (\ref{finriskbound}). Due to the property of the typical set, this term decreases
to zero as $n\rightarrow \infty$.
Therefore, the first term is the main term, which has a form of
redundancy of two-stage code like the quantized case. Hence, this theorem gives a
justification of the MDL principle in supervised learning.
Note that, however, 
$-\log p_2(y^n|x^n)$ needs to satisfy Kraft's inequality in order to
interpret the main term as a conditional redundancy exactly.
A sufficient conditions for this was introduced by \cite{Chatterjee2014} and is called
`codelength validity'.
\begin{definition}[codelength validity]\label{codelengthvalid}
We say that $\copen(\theta|x^n)$ is codelength valid if
 there exist a quantized subset $\cTheta(x^n) \subset \Theta$ and a
 model description length $\ctpen(\ctheta|x^n)$ satisfying Kraft's inequality such that 
\begin{align} 
\forall y^n\in \rs{Y}^n,\ &
\max_{\theta \in \Theta}
\Bigl\{
-
\log \frac{p_*(y^n|x^n)}{p_{\theta}(y^n|x^n)}
- \copen(\theta|x^n)
\Bigr\}\nonumber\\
 \le
&\max_{\ctheta \in \cTheta(x^n)}
\Bigl\{
-
\log \frac{p_*(y^n|x^n)}{p_{\ctheta}(y^n|x^n)}
- \ctpen(\ctheta|x^n)
\Bigr\} \label{cvcond}
\end{align}
for each $x^n$. 
\end{definition}
We note that both the quantization and the model description length on
it depend on $x^n$ in contrast to the $\epsilon$-risk validity.
This is because the fixed design setting suffices to justify the
redundancy interpretation.
Let us see that $-\log p_2(y|x)$ can be exactly interpreted as a codelength if
$\copen(\theta|x^n)$ is codelength valid.
First, we assume that $\rs{Y}$, the range of $y$, is discrete. 
For each $x^n$, we have
 \begin{eqnarray*}
  \lefteqn{\sum_{y^n\in \rs{Y}^n}\exp\left(-(-\log \twopc(y^n|x^n))\right)}\\
   &=&\sum_{y^n}\exp\left(\max_{\theta\in \Theta}\left\{\log p_{\theta}(y^n|x^n)-\copen(\theta|x^n)\right\}\right)\\
   &\le&\sum_{y^n}\exp\left(\max_{\ctheta\in \cTheta(x^n)}\left\{\log p_{\ctheta}(y^n|x^n)-\ctpen(\ctheta|x^n)\right\}\right)\\
   &\le&\sum_{y^n}\sum_{\ctheta\in \cTheta(x^n)}\exp\left(\log p_{\ctheta}(y^n|x^n)-\ctpen(\ctheta|x^n)\right)\\
   &=&\sum_{\ctheta\in
    \cTheta(x^n)}\exp\left(-\ctpen(\ctheta|x^n)\right)\sum_{y^n}p_{\ctheta}(y^n|x^n)\le 1.
 \end{eqnarray*}
Hence, $-\log p_2(y^n|x^n)$ can be exactly interpreted as
  a codelength of a prefix code. Next, we consider the case where $\rs{Y}$ is a
  continuous space. The above inequality trivially holds by replacing the
  sum with respect to $y^n$ with an integral. Thus, $p_2(y^n|x^n)$ is
  guaranteed to be a sub-probability density function.
  Needless to say, $-\log p_2(y^n|x^n)$ cannot be interpreted as a
  codelength  as itself in continuous cases. As is well
  known, however, a difference $(-\log p_2(y^n|x^n))-(-\log
  p_*(y^n|x^n))$ can be exactly interpreted as a codelength difference
  by way of quantization. See Section III of \cite{Barron1991} for details.
This indicates that both the redundancy interpretation of the fist term
  of (\ref{finriskbound}) and the regret interpretation of the
  (negative) second term in the left side of the inequality in the first
  line of (\ref{regretbound}) are justified by the codelength validity. 
Note that, however, the $\epsilon$-risk validity does not imply the
  codelength validity and vice versa in general.

We discuss about the conditional expectation in the risk bound
(\ref{finriskbound}). This conditional expectation seems to be hard to
be replaced with the usual (unconditional) expectation. The main
difficulty arises from the unboundedness of the loss function. Indeed, we
can immediately show a similar risk bound with unconditional expectation for
bounded loss functions. As an example, let
us consider a class of divergence, called $\alpha$-divergence
\cite{Cichocki2010}
\begin{eqnarray}
 &&\hskip-4mm\rs{D}^n_{\alpha}(p,r):=\nonumber\\
 &&\hskip-4mm\frac{4}{1-\alpha^2}\int\bigg(1-\left(\frac{r(y^n|x^n)}{p(y^n|x^n)}\right)^{\frac{1+\alpha}{2}}\bigg)q_*(x^n)p(y^n|x^n)dx^ndy^n.\nonumber\\\label{alphadiv}
\end{eqnarray}
The $\alpha$-divergence approaches KL divergence as $\alpha\rightarrow
\pm 1$ \cite{amarinagaoka00}. More exactly,
\begin{equation}
 \lim_{\alpha\rightarrow -1} \rs{D}^n_{\alpha}(p,r)=\rs{D}^n(p,r),\quad
 \lim_{\alpha\rightarrow 1} \rs{D}^n_{\alpha}(p,r)=\rs{D}^n(r,p). \label{alphaKL}
\end{equation}
We also note that the $\alpha$-divergence with $\alpha=0$ is four times the
 squared Hellinger distance
\begin{eqnarray}
 \lefteqn{d^{2,n}_H(p,r)=}\nonumber \\
&&\hskip-1.3cm  \int\!\!\!
 \left(\sqrt{p(y^n|x^n)}-\sqrt{r(y^n|x^n)}\right)^2\!\!q_*(x^n)p(y^n|x^n)dx^ndy^n,
 \label{hellinger}
\end{eqnarray}
 which has been studied and used in statistics for a long time.
 We focus here on the following two properties of $\alpha$-divergence:
\begin{itemize}
 \item[(i)] The $\alpha$-divergence is always bounded:
\begin{equation}
\rs{D}^n_{\alpha}(p,r)\in [0,4/(1-\alpha^2)] \label{boundedalpha}
\end{equation}
	    for any $p,r$ and $\alpha\in (-1,1)$.
 \item[(ii)] The $\alpha$-divergence is bounded by the \renyi divergence as
\begin{equation}
d^n_{(1-\alpha)/2}(p,r)\ge 
\frac{1-\alpha}{2}\rs{D}^n_{\alpha}(p,r)  \label{alpharenyi}
\end{equation}
for any $p,r$ and $\alpha\in (-1,1)$. See \cite{takeuchi14} for its
	     proof.
\end{itemize}
As a corollary of Theorem \ref{mytheorem1}, we obtain the
following risk bound.
 \begin{corollary}\label{coro1}
 Let $\beta,\epsilon$ be arbitrary real numbers in $(0,1)$.
Define a function $\lambda(t):=(1-t)/2$. 
 For any $\alpha\in [2\beta-1,1)$, if $\copen(\theta|x^n)$ is
 $\epsilon$-risk valid for $(\lambda(\alpha),\beta,\mathcal{P}^n_x,\typicalset)$ and
 $p_2(y^n|x^n)$ is a sub-probability distribution,
\begin{eqnarray*}
  E_{\bar{p}_*}[\rs{D}^n_{\alpha}(p_*,p_{\hat{\theta}})]&\le&
   \frac{1}{\lambda(\alpha)}
   E_{\bar{p}_*}\left[
     \log \frac{p_*(y^n|x^n)}{\twopc(y^n|x^n)}
    \right]\\
 &&+ \frac{\inprob}{\lambda(\alpha)\beta}\log \frac{1}{\inprob }
  +\frac{(1-\inprob)}{\lambda(\alpha)(\lambda(\alpha)+\alpha)},
\end{eqnarray*}
 In particular, taking $\beta =(\alpha+1)/2$ yields the tightest bound
 \begin{eqnarray}
  E_{\bar{p}_*}[\rs{D}^n_{\alpha}(p_*,p_{\hat{\theta}})]&\le&
   \frac{1}{\lambda(\alpha)}
   E_{\bar{p}_*}\left[
     \log \frac{p_*(y^n|x^n)}{\twopc(y^n|x^n)}
    \right]\nonumber \\
 &&+ \frac{\inprob}{\lambda(\alpha)(\lambda(\alpha)+\alpha)}\log \frac{1}{\inprob }\nonumber\\
&& +\frac{(1-\inprob)}{\lambda(\alpha)(\lambda(\alpha)+\alpha)}.\label{alphabound}
\end{eqnarray}
 \end{corollary}
 Its proof will be described in Section \ref{coro1proof}.
Though it is not so obvious when the condition ``$p_2(y^n|x^n)$ is a
 sub-probability distribution'' is satisfied, we remark that the codelength validity of
 $\copen(\theta|x^n)$ is its simple sufficient condition. The second and
the third terms of the right side vanish as $n\rightarrow \infty$ due
to the property of the typical set. The boundedness of loss function is indispensable for the proof. On the other hand,
it seems to be impossible to bound the risk for unbounded loss
functions. Our remedy for this issue is the risk evaluation based on the
conditional expectation on the typical set. Because $x^n$ lies out of
$\typicalset$ with small probability, the conditional expectation is
likely to capture the expectation of almost all cases.  
In spite of this fact, if one wants to remove the unnatural conditional
expectation, Theorem \ref{mytheorem2} offers a more satisfactory
bound. Note that the right side of \eqnum{regretbound} also approaches
to zero as $n\rightarrow \infty$.

We remark the relationship of our result with KL
divergence $\rs{D}^n(p,r)$. Because of (\ref{renyiKL}) or
(\ref{alphaKL}), it seems to be possible to obtain a risk bound with KL
divergence. However, it is impossible because taking $\lambda\rightarrow
1$ in (\ref{finriskbound}) or $\alpha\rightarrow \pm 1$ in
(\ref{alphabound}) makes the bounds diverge to the infinity.
That is, we cannot derive a risk bound for the risk with KL divergence by BC
theory, though we can do it for the \renyi divergence and the
$\alpha$-divergence.
It sounds somewhat strange because KL divergence seems to be
related the most to the notion of the MDL principle because it has a
clear information theoretical interpretation.
This issue originates from the original BC theory and has been casted as
an open problem for a long time.

Finally, we remark that the effectiveness of our proposal in real
situations depends on whether we can show the risk validity of the
target penalty and derive a sufficiently small bound for $\log(1/\inprob)$ and $1-\inprob$. Actually, much effort is required to realize them for lasso. 
\subsection{Risk Bound of Lasso in Random Design}\label{lasso}
In this section, we apply the approach in the previous section to lasso
and derive new risk and regret bounds. In a setting of lasso, training data
$\{(x_i,y_i)\in \Re^p\times \Re|i=1,2,\cdots, n\}$ obey a usual regression
model $y_i=x_i^T\theta^*+\epsilon_i$ for $i=1,2,\cdots, n$, where
$\theta^*$ is a true parameter and $\epsilon_i$ is a Gaussian
noise having zero mean and a known variance $\sigma^2$. By introducing
$Y:=(y_1,y_2,\cdots, y^n)^T$, $\rs{E}:=(\epsilon_1,\epsilon_2,\cdots,
\epsilon_n)^T$ and an $n\times p$ matrix $X:=[x_1\ x_2\ \cdots x_n]^T$, we 
have a vector/matrix expression of the regression model
$Y=X\theta^*+\rs{E}$. The parameter space $\Theta$ is
$\Re^p$.
The dimension $p$ of parameter
$\theta$ can be greater than $n$.
 The lasso estimator is defined by
\begin{equation}
  \hat{\theta}(x^n,y^n):=\arg \min_{\theta\in
   \Theta}\left\{\frac{1}{2n\sigma^2}\|Y-X\theta\|_2^2+\mu_1\|\theta\|_{w,1}\right\}, \label{lassoestimator}
\end{equation}
where $\mu_1$ is a positive real number (regularization coefficient).
Note that the weighted $\ell_1$ norm is used in (\ref{lassoestimator}), 
though the original lasso was defined with the usual $\ell_1$ norm in
\cite{tibshirani96a}. As explained in Section \ref{bctheory},
$\hat{\theta}$ corresponds to the usual lasso with `column
normalization'. When $x^n$ is Gaussian with zero mean, we can derive a risk valid weighted
$\ell_1$ penalty by choosing an appropriate typical set. 
    \begin{lemma}\label{erv4lasso}
     For any $\epsilon\in (0,1)$, define
      \begin{eqnarray}
       \mathcal{P}^n_x\!\!\!&:=&\!\!\!\{q(x^n)=\Pi_{i=1}^nN(x_i|\mib{0},\Sigma)| \mbox{
	non-singular }\Sigma\},\nonumber \\
       \typicalset\!\!\!&:=&\!\!\!\Bigl\{x^n\,\Bigl|\,\forall j, 1-\epsilon\le \frac{(1/n)\sum_{i=1}^nx_{ij}^2}{\Sigma_{jj}}\le 1+\epsilon\Bigr\},\label{tset}
      \end{eqnarray}
where $N(x|\mu,\Sigma)$ is a Gaussian distribution with a mean vector $\mu$ and a
  covariance matrix $\Sigma$. Here, $\Sigma_{jj}$ denotes the
    $j$th diagonal element of $\Sigma$
    and $x_{ij}$ denotes the $j$th
     element of $x_i$.
     Assume a linear regression setting:
     \begin{eqnarray*}
      p_*(y^n|x^n)&=&\Pi_{i=1}^nN(y_i|x_i^T\theta^*,\sigma^2),     \\
     p_{\theta}(y^n|x^n)&=&\Pi_{i=1}^nN(y_i|x_i^T\theta,\sigma^2).     
     \end{eqnarray*}
Let $\beta$ be a real number in $(0,1)$ and $\lambda$ be a real number
  in $(0,1-\beta]$.
    The weighted $\ell_1$ penalty $\copen_1(\theta|x^n)=\mu_1\|\theta\|_{w,1}+\mu_2$ is $\epsilon$-risk
  valid for $(\lambda,\beta,\mathcal{P}^n_x,\typicalset)$ if
  \begin{eqnarray}
\mu_1\ge 
   \sqrt{
 \frac{n\log 4p}{\beta\sigma^2(1-\epsilon)}\cdot
\frac{\lambda+8\sqrt{1-\epsilon^2}}{4}
},\quad \mu_2\ge \frac{\log 2}{\beta}. \label{rvcondition}
  \end{eqnarray}
   \end{lemma}
   We describe its proof in Section \ref{erv4lassoproof}.
   The derivation
is much more complicated and requires more techniques, compared to the
fixed design case in \cite{chatterjeebarron14}. This is because the
\renyi divergence is a usual mean square error (MSE) in the fixed design
case, while it is not in the random design case in general. In addition,
it is important for the risk bound derivation to choose an appropriate
typical set in a sense that we can show that $\inprob$ approaches to one
sufficiently fast and we can also show the $\epsilon$-risk validity of the target
penalty with the chosen typical set. In case of lasso with normal
design, the typical set $\typicalset$ defined in (\ref{tset}) satisfies
such properties.

Let us compare the coefficient of the risk valid weighted $\ell_1$ penalty with the
fixed design case in \cite{chatterjeebarron14}.
They showed that the weighted $\ell_1$ norm satisfying
\begin{equation}
 \mu_1\ge \sqrt{\frac{2n\log 4p}{\sigma^2}},\quad \mu_2\ge \frac{\log
  2}{\beta} \label{riskvalidfixed}
\end{equation}
is risk valid in the fixed design case. The condition for $\mu_2$ is the
same, while the condition for $\mu_1$ in (\ref{rvcondition}) is more
strict than that of the fixed design case.
We compare them by taking $\beta=1-\lambda$ (the tightest choice) and
$\epsilon=0$ in (\ref{rvcondition}) because $\epsilon$ can be
negligibly small for sufficiently large $n$. The minimum $\mu_1$ for the risk validity in
the random design case is  
\[
 \sqrt{\frac{\lambda+8}{8(1-\lambda)}}
\]
times that for the fixed design case.
Hence, the smallest value of regularization coefficient $\mu_1$ for
which the risk bound holds in the random design is always larger than that of
the fixed design case for any $\lambda\in (0,1)$ but its extent is not
so large unless $\lambda$ is extremely close to $1$ (See Fig. \ref{compfixed}).
\begin{figure}
\centering
 \includegraphics[width=2.7in]{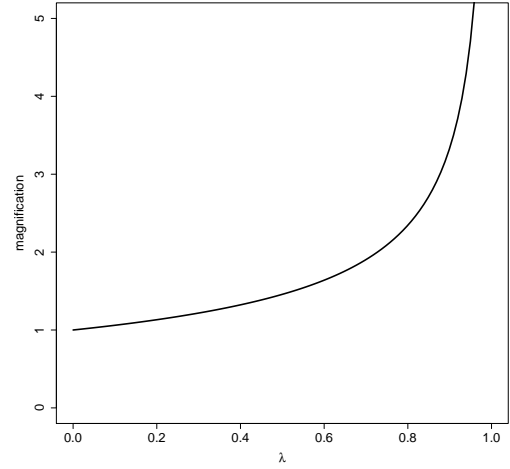}
 \caption{Plot of $\sqrt{(\lambda+8)/8(1-\lambda)}$ against $\lambda$.}
\label{compfixed}
\end{figure}

Next, we show that $\inprob$ exponentially approaches to one
as $n$ increases. 
 \begin{lemma}[Exponential Bound of Typical Set]\label{expbound}
  Suppose that $x_i\sim N(x_i|0,\Sigma)$ independently.
For any $\epsilon\in (0,1)$,
\begin{eqnarray}
\inprob&\ge&\left(1-\!2\exp\left(-\frac{n}{2}(\epsilon-\log(1+\epsilon))\right)\right)^p\label{expb}\\
&\ge&
 1-2p\exp\left(-\frac{n}{2}(\epsilon-\log(1+\epsilon))\right)\nonumber\\
 &\ge& 1-2p\exp\left(\!-\frac{n\epsilon^2}{7}\right). \nonumber
\end{eqnarray}
 \end{lemma}
 See Section \ref{expboundproof} for its proof.
 In the lasso case, it is often
postulated that $p$ is much greater than $n$. Due to Lemma \ref{expbound},
$1-\inprob$ is $O(p\cdot \exp(-n\epsilon^2/7))$, which also implies that the second
term in \eqnum{finriskbound} can be negligibly small even if $n\ll p$. 
In this sense, the exponential bound is important for lasso.
Combining Lemmas \ref{erv4lasso} and \ref{expbound} with Theorems \ref{mytheorem1}
 and \ref{mytheorem2}, we obtain the following theorem.
\begin{theorem}\label{theorem4lasso}
 For any $\epsilon\in (0,1)$, define
\begin{eqnarray}
 \mathcal{P}^n_x\!\!\!&:=&\!\!\!\{q(x^n)=\Pi_{i=1}^nN(x_i|\mib{0},\Sigma)| \mbox{
 non-singular }\Sigma\},\nonumber \\
 \typicalset\!\!\!&:=&\!\!\!\Bigl\{x^n\,\Bigl|\,\forall j, 1-\epsilon\le \frac{(1/n)\sum_{i=1}^nx_{ij}^2}{\Sigma_{jj}}\le 1+\epsilon\Bigr\}.\nonumber
\end{eqnarray}
     Assume a linear regression setting:
     \begin{eqnarray*}
      p_*(y^n|x^n)&=&\Pi_{i=1}^nN(y_i|x_i^T\theta^*,\sigma^2),     \\
     p_{\theta}(y^n|x^n)&=&\Pi_{i=1}^nN(y_i|x_i^T\theta,\sigma^2).     
     \end{eqnarray*}
 Let $\beta$ be a real number in $(0,1)$. 
 For any $\lambda\in (0,1-\beta]$, if 
  \begin{eqnarray}
\mu_1\ge 
   \sqrt{
 \frac{\log 4p}{n\beta\sigma^2(1-\epsilon)}\cdot
\frac{\lambda+8\sqrt{1-\epsilon^2}}{4}
},\quad \mu_2\ge \frac{\log 2}{n\beta}, \nonumber
  \end{eqnarray}
the lasso estimator $\hat{\theta}(x^n,y^n)$ in (\ref{lassoestimator}) has a risk bound 
\begin{eqnarray}
\lefteqn{\condexp [d_{\lambda}(p_*,p_{\hat{\theta}(x^n,y^n)})]\le }\nonumber\\
&& \hskip-7mm\condexp \Biggl[\inf_{\theta\in
  \Theta}\bigg\{\frac{\left(\|Y-X\theta\|_2^2-\|Y-X\theta^*\|_2^2\right)}{2n\sigma^2}
  +\mu_1\|\theta\|_{w,1}+\mu_2\bigg\}\Biggr]\nonumber\\
 &&\hskip-7mm-\frac{p\log\left(1\!-\!2\exp\left(-\frac{n}{2}(\epsilon-\log(1+\epsilon))\right)\right)}{n\beta},\label{riskboundlasso}
\end{eqnarray}
 and a regret bound
 \begin{eqnarray}
 \lefteqn{d_{\lambda}(p_*,p_{\hat{\theta}(x^n,y^n)})\le} \nonumber\\
 &&\hskip-7mm
\inf_{\theta\in
  \Theta}\bigg\{\frac{\left(\|Y-X\theta\|_2^2-\|Y-X\theta^*\|_2^2\right)}{2n\sigma^2}
  +\mu_1\|\theta\|_{w,1}+\mu_2\bigg\}+\tau\nonumber\\ \label{regretboundlasso}
 \end{eqnarray}
 with probability at least
 \begin{equation}
 \left(1-\!2\exp\left(-\frac{n}{2}(\epsilon-\log(1+\epsilon))\right)\right)^p-\exp(-\tau
 n\beta), \label{lowerboundP}
 \end{equation}
 which is bounded below by
 \[
  1-O\left(p\cdot \exp\left(-n\kappa\right)\right)
 \]
 with $\kappa:=\min\{\epsilon^2/7,\tau\beta\}$. 
\end{theorem}
Since $x^n$ and $y^n$ are i.i.d. now,
$d^n_{\lambda}(p,r)=nd_{\lambda}(p,r)$. Hence, we presented the risk
bound as a single-sample version in (\ref{riskboundlasso}) by dividing
the both sides by $n$. 
Finally, we remark that the following interesting fact holds for the lasso case.
  \begin{lemma}\label{cv4lasso}
     Assume a linear regression setting:
     \begin{eqnarray*}
      p_*(y^n|x^n)&=&\Pi_{i=1}^nN(y_i|x_i^T\theta^*,\sigma^2),     \\
     p_{\theta}(y^n|x^n)&=&\Pi_{i=1}^nN(y_i|x_i^T\theta,\sigma^2).     
     \end{eqnarray*}
If $\mu_1$ and $\mu_2$ satisfy (\ref{rvcondition}), then the weighted
 $\ell_1$ norm $L(\theta|x^n)=\mu_1\|\theta\|_{w,1}+\mu_2$ is codelength valid. 
\end{lemma}
That is, the weighted $\ell_1$ penalties derived in Lemma
\ref{erv4lasso} are not only $\epsilon$-risk valid but also codelength
valid. Its proof will be described in Section \ref{proofcv4lasso}.
By this fact, the redundancy and regret interpretation of the main terms in (\ref{riskboundlasso}) and (\ref{regretboundlasso}) are justified. 
It also indicates that we can obtain the unconditional risk bound with
respect to $\alpha$-divergence for those weighted $\ell_1$ penalties by Corollary
\ref{coro1} without any additional condition.
\section{Proofs of Theorems, Lemmas and Corollary}\label{proofs}
We give all proofs to the theorems, the lemmas and the
corollary in the previous section. 
\subsection{Proof of Theorem \ref{mytheorem1}}\label{mytheorem1proof}
Here, we prove our main theorem. The proof proceeds along with
the same line as \cite{chatterjeebarron14} though some modifications
are necessary.
\begin{proof}
 Define
 \[
  F_{\lambda}^{\theta}(x^n,y^n):=d^n_{\lambda}(p_*,p_{\theta})-\log \frac{p_*(y^n|x^n)}{p_{\theta}(y^n|x^n)}.
 \]
By the $\epsilon$-risk validity, we obtain
 \begin{eqnarray}
\lefteqn{\condexp \Big[\exp\Big(\beta\max_{\theta\in
 \Theta}\Big\{F_{\lambda}^{\theta}(x^n,y^n)-\copen(\theta|x^n)\Big\}\Big)\Big]}\nonumber\\
&\hspace{-9mm}\le &\hspace{-5mm}\condexp \Big[\exp\Big(\beta\max_{\ctheta \in
 \cTheta}\Big\{F_{\lambda}^{\ctheta}(x^n,y^n)-\ctpen(\ctheta |q_*)\Big\}\Big)\Big]\nonumber\\
&\hspace{-9mm}\le &\hspace{-5mm}\sum_{\ctheta \in \cTheta(q_*)}\condexp \Big[\exp\Big(\beta\Big(F_{\lambda}^{\ctheta}(x^n,y^n)-\ctpen(\ctheta |q_*)\Big)\Big)\Big]\nonumber\\
&\hspace{-9mm}= &\hspace{-5mm}\sum_{\ctheta \in
 \cTheta(q_*)}\exp(-\beta\ctpen\big(\ctheta |q_*)\big)\condexp \Big[\exp\Big(\beta
 F_{\lambda}^{\ctheta}(x^n,y^n)\Big)\Big].\label{intermediate4}
  \end{eqnarray}
  The following fact is an extension of the key technique of BC theory:
  \begin{eqnarray*}
\lefteqn{\condexp \Bigg[\exp\Big(\beta F_{\lambda}^{\ctheta}(x^n,y^n)\Big)\Bigg]}\\
 &=&\exp\left(\beta d^n_{\lambda}(p_*,p_{\ctheta})\right)
 \condexp \left[\left(\frac{p_{\ctheta }(y^n|x^n)}{p_*(y^n|x^n)}\right)^{\beta}\right]\\
 &\le&\frac{1}{\inprob}\exp\left(\beta d^n_{\lambda}(p_*,p_{\ctheta})\right)
 E_{\bar{p}_*} \left[\left(\frac{p_{\ctheta }(y^n|x^n)}{p_*(y^n|x^n)}\right)^{\beta}\right]\\
 &=&\frac{1}{\inprob}\exp\left(\beta
			    d^n_{\lambda}(p_*,p_{\ctheta})\right)\exp\left(-\beta
			    d_{1-\beta}^n(p_*,p_{\ctheta})\right)
 \\
 &\le &\frac{1}{\inprob}\exp\left(\beta
			       d^n_{\lambda}(p_*,p_{\ctheta})\right)\exp\left(-\beta
			       d^n_{\lambda}(p_*,p_{\ctheta})\right)=\frac{1}{\inprob}.
  \end{eqnarray*}
 The first inequality holds because $E_{\bar{p}_*(x^n,y^n)}\left[A\right]\ge
  \inprob \condexp \left[A\right]$ for any non-negative random variable
 $A$. The second inequality holds because of the monotonically
 increasing property of $d^n_{\lambda}(p_*,p_{\theta})$ in terms of $\lambda$. Thus, the right
 side of (\ref{intermediate4}) is bounded as
\begin{eqnarray*}
&&  \sum_{\ctheta \in
 \cTheta(q_*)}\exp(-\beta\ctpen\big(\ctheta |q_*)\big)\condexp \Big[\exp\Big(\beta
 F_{\lambda}^{\ctheta}(x^n,y^n)\Big)\Big]\\
 & &\le
\frac{1}{\inprob} \sum_{\ctheta \in
 \cTheta(q_*)}\exp(-\beta\ctpen\big(\ctheta |q_*)\big)\le \frac{1}{\inprob}.
\end{eqnarray*}
Hence, we have an important inequality
\begin{eqnarray}
\hskip-5mm\frac{1}{\inprob}
\!\!&\!\!\ge\!\!& \!\!\condexp \left[\exp\left(\beta\max_{\theta\in
 \Theta}\left\{F_{\lambda}^{\theta}(x^n,y^n)-\copen(\theta|x^n)\right\}\right)\right].\label{coroused}
\end{eqnarray}
Applying Jensen's inequality to (\ref{coroused}), we have
 \begin{eqnarray}
\hskip-3mm\frac{1}{\inprob}
\!\!\!\!\!
&\ge& \!\!\!\!\exp\left(\condexp \left[\beta\max_{\theta\in
 \Theta}\left\{F_{\lambda}^{\theta}(x^n,y^n)-\copen(\theta|x^n)\right\}\right]\right)\nonumber \\
&\ge& \!\!\!\!\exp\left(\condexp \left[\beta\left(F_{\lambda}^{\hat{\theta}}(x^n,y^n)-\copen(\hat{\theta}|x^n)\right)\right]\right).\nonumber
 \end{eqnarray}
Thus, we have
\begin{eqnarray*}
 -\frac{\log \inprob}{\beta } \ge  \condexp \left[d^n_{\lambda}(p_*,p_{\hat{\theta}})-\log\frac{p_*(y^n|x^n)}{p_{\hat{\theta}}(y^n|x^n)}-\copen(\hat{\theta}|x^n)\right].
\end{eqnarray*}
Rearranging the terms of this inequality, we have the statement. 
\end{proof}

\subsection{Proof of Theorem \ref{mytheorem2}}\label{mytheorem2proof}
It is not necessary to start from scratch. We reuse the proof of Theorem \ref{mytheorem1}.
 \begin{proof}
We can start from \eqnumn{coroused}. 
  For convenience, we define
  \begin{eqnarray*}
   \lefteqn{
    \xi(x^n,y^n)}\\
   &=&\frac{1}{n}\max_{\theta\in
  \Theta}\left\{F_{\lambda}^{\theta}(x^n,y^n)-\copen(\theta|x^n)\right\}\\
 &=&\max_{\theta\in \cTheta}\left\{\frac{d^n_{\lambda}(p_*,p_{\theta})}{n}-\frac{1}{n}\log \frac{p_*(y^n|x^n)}{p_{\theta}(y^n|x^n)}-\frac{\copen(\theta|x^n)}{n}\right\}.
  \end{eqnarray*}
By Markov's inequality and (\ref{coroused}), 
\begin{eqnarray*}
 \lefteqn{\Pr\left(\xi(x^n,y^n)\ge \tau|x^n\in \typicalset\right)}\\
 &=&
\Pr\left(\exp\left(n\beta\xi(x^n,y^n)\right)\ge \exp(n\beta\tau)|x^n\in
   \typicalset\right) \\
 &\le& \frac{\exp(-n\tau\beta)}{\inprob}.
\end{eqnarray*}
Hence, we obtain 
\begin{eqnarray}
\lefteqn{\Pr\left(\xi(x^n,y^n)\ge \tau\right)}\nonumber\\
&=&\!\!\!\!
\inprob \Pr\left(\xi(x^n,y^n)\ge \tau |x^n\in \typicalset\right)\nonumber\\
&&+(1-\inprob )\Pr\left(\xi(x^n,y^n)\ge \tau |x^n\notin \typicalset\right)\nonumber\\
&\le &\!\!\!\!
\inprob \Pr\left(\xi(x^n,y^n)\ge \tau |x^n\in \typicalset\right)+(1-\inprob )\nonumber\\
&\le &\!\!\!\!
\exp(-n\tau\beta)+(1-\inprob ).\nonumber
\end{eqnarray}
The proof completes by noticing that $(1/n)\left(F^{\hat{\theta}}_{\lambda}(x^n,y^n)-L(\hat{\theta}|x^n)\right)\le \xi(x^n,y^n)$
  for any $x^n$ and $y^n$.
 \end{proof}

 \subsection{Proof of Corollary \ref{coro1}}\label{coro1proof}
The proof is obtained immediately from Theorem \ref{mytheorem1}.
   \begin{proof}
Let again $\condexp $ denote a conditional expectation with
    regard to $\bar{p}_*(x^n,y^n)$ given that $x^n\in
   \typicalset$. Let further $I_A(x^n)$ be an indicator function of a set
    $A\subset \rs{X}^n$. 
    The unconditional risk is bounded as
  \begin{eqnarray*}
   \lefteqn{E_{\bar{p}_*}[\rs{D}^n_{\alpha}(p_*,p_{\hat{\theta}})]}\\
   &&\hskip-7mm
    =E_{\bar{p}_*}[I_{\typicalset}(x^n)\rs{D}^n_{\alpha}(p_*,p_{\hat{\theta}})]
     \!+\!E_{\bar{p}_*}[(1-I_{\typicalset}(x^n))\rs{D}^n_{\alpha}(p_*,p_{\hat{\theta}})]
     \\
   &&\hskip-7mm\le \inprob \condexp [\rs{D}^n_{\alpha}(p_*,p_{\hat{\theta}})]
    +(1-\inprob)\cdot \frac{4}{1-\alpha^2}\\
   &&\hskip-7mm\le \frac{\inprob}{\lambda(\alpha)}\condexp [d^n_{\lambda(\alpha)}(p_*,p_{\hat{\theta}})]
    +\frac{(1-\inprob)}{\lambda(\alpha)(\lambda(\alpha)+\alpha)}\\
   &&\hskip-7mm\le \frac{\inprob}{\lambda(\alpha)}\left(\condexp 
\log \frac{p_*(y^n|x^n)}{\twopc(y^n|x^n)}
+ \frac{1}{\beta} \log \frac{1}{\inprob }
					       \right)\\
   &&+\frac{(1-\inprob)}{\lambda(\alpha)(\lambda(\alpha)+\alpha)}\\
   &&\hskip-7mm= \frac{1}{\lambda(\alpha)}E_{\bar{p}_*}\left[
I_{\typicalset}(x^n)
\log \frac{p_*(y^n|x^n)}{\twopc(y^n|x^n)}\right]
   + \frac{\inprob}{\lambda(\alpha)\beta}\log \frac{1}{\inprob }\\
&&
    +\frac{(1-\inprob)}{\lambda(\alpha)(\lambda(\alpha)+\alpha)}\\
   &&\hskip-7mm\le \frac{1}{\lambda(\alpha)}E_{\bar{p}_*}\left[
					  \log \frac{p_*(y^n|x^n)}{\twopc(y^n|x^n)}\right]
   + \frac{\inprob}{\lambda(\alpha)\beta}\log \frac{1}{\inprob }\\
   &&+\frac{(1-\inprob)}{\lambda(\alpha)(\lambda(\alpha)+\alpha)}.
  \end{eqnarray*}
    The first and second inequalities follow from the two
    properties of $\alpha$-divergence in (\ref{boundedalpha}) and (\ref{alpharenyi}) respectively. The third inequality 
    follows from Theorem \ref{mytheorem1} because $\lambda(\alpha)\in
    (0,1-\beta)$ by the assumption. 
The last inequality holds because of the following reason. By the
    decomposition of expectation, we have
    \begin{eqnarray*}
\lefteqn{E_{\bar{p}_*(x^n,y^n)}\Big[
I_{\typicalset}(x^n)\log \frac{p_*(y^n|x^n)}{\twopc(y^n|x^n)}\Big]}\\
&&\hskip-7mm=
E_{q_*(x^n)}\bigg[
I_{\typicalset}(x^n)E_{p_*(y^n|x^n)}\bigg[
\log \frac{p_*(y^n|x^n)}{\twopc (y^n|x^n)}\bigg]\bigg].
    \end{eqnarray*}
    Since $\twopc(y^n|x^n)$ is a sub-probability distribution by the
    assumption, the conditional expectation part is non-negative. Therefore, removing the indicator
    function $I_{\typicalset}(x^n)$ cannot decrease this quantity.
    The
    final part of the statement follows from the fact that taking
    $\lambda=1-\beta$ makes the bound in (\ref{finriskbound}) tightest because of the monotonically increasing property of \renyi
    divergence with regard to $\lambda$.
   \end{proof}
Again, we remark that the sub-probability condition of $p_2(y^n|x^n)$ can
be replaced with a sufficient condition ``$L(\theta|x^n)$ is codelength
valid.''  
In addition, the sub-probability condition can be relaxed to 
\[
\sup_{x^n\in \rs{X}^n}\int p_2(y^n|x^n)dy^n < \infty,
\]
under which the bound increases by 
$(1-\inprob)\log \sup_{x^n\in \rs{X}^n} \int p_2(y^n|x^n)dy^n$.
\subsection{\renyi Divergence and Its Derivatives}\label{bqlemmasec}
In this section and the next section, we prove a series of lemmas, which
will be used to derive risk valid penalties for lasso. 
First, we show that the \renyi divergence can be understood by
defining $\bar{p}^{\lambda}_{\theta}(x,y)$ in Lemma \ref{renyiview}. Then,  
their explicit forms in the lasso setting are calculated in Lemma
\ref{bqlemma}. 
  \begin{lemma}\label{renyiview}
   Define a probability distribution $\bar{p}^{\lambda}_{\theta}(x,y)$ by
  \begin{eqnarray*}
 \bar{p}^{\lambda}_{\theta}(x,y):= \frac{q_*(x)p_*(y|x)^{\lambda}p_{\theta}(y|x)^{1-\lambda}}{Z^{\lambda}_{\theta}},
  \end{eqnarray*}
where $Z^{\lambda}_{\theta}$ is a normalization constant. 
Then, the \renyi divergence and its first and
second derivatives are written as
\begin{eqnarray}
 d_{\lambda}(p_*,p_{\theta})&=&\frac{-1}{1-\lambda}\log
  Z^{\lambda}_{\theta},\nonumber\\
  \frac{\partial d_{\lambda}(p_*,p_{\theta})}{\partial
  \theta}&=&-E_{\bar{p}^{\lambda}_{\theta}}\left[s_{\theta}(y|x)\right],\label{firstderivative}\\
 \frac{\partial^2d_{\lambda}(p_*,p_{\theta})}{\partial \theta\partial
  \theta^T}\!\!&=&\!\!-E_{\bar{p}^{\lambda}_{\theta}}\left[
G_{\theta}(x,y)
\right]\nonumber\\
&&-(1-\lambda){\rm Var}_{\bar{p}^{\lambda}_{\theta}}\left(s_{\theta}(y|x)\right),\label{hessiangeneral}
\end{eqnarray}
where ${\rm Var}_p(A)$ denotes a covariance matrix of $A$ with respect to $p$ and 
\begin{eqnarray*}
s_{\theta}(y|x)&:=&\frac{\partial \log
     p_{\theta}(y|x)}{\partial \theta}, \\ 
G_{\theta}(x,y)&:=&\frac{\partial^2\log
 p_{\theta}(y|x)}{\partial \theta\partial \theta^T}.
\end{eqnarray*}
\end{lemma}

 \begin{proof}
The normalizing constant is rewritten as
\begin{eqnarray*}
 Z^{\lambda}_{\theta}
  &=&\int q_*(x)p_*(y|x)\left(\frac{p_{\theta}(y|x)}{p_*(y|x)}\right)^{1-\lambda}dxdy\\
&=&E_{\bar{p}_*(x,y)}\left[ \left(\frac{p_{\theta}(y|x)}{p_*(y|x)}\right)^{1-\lambda}\right].
\end{eqnarray*}
Thus, the
 \renyi divergence is written as
 \[
 d_{\lambda}(p_*,p_{\theta})=-\frac{1}{1-\lambda}\log Z^{\lambda}_{\theta}.
 \]
  Next, we calculate the partial derivative of $\log
  Z^{\lambda}_{\theta}$ as
  \begin{eqnarray*}
\lefteqn{\frac{\partial \log Z^{\lambda}_{\theta}}{\partial \theta}}\\
    &=&
    \frac{1}{Z^{\lambda}_{\theta}}\frac{\partial
    Z^{\lambda}_{\theta}}{\partial \theta}\\
 &=&    \frac{1}{Z^{\lambda}_{\theta}}
  E_{\bar{p}_*}\left[\left(\frac{p_{\theta}(y|x)}{p_*(y|x)}\right)^{1-\lambda}\frac{\partial
				 }{\partial
				 \theta}\log
				 \left(\frac{p_{\theta}(y|x)}{p_*(y|x)}\right)^{1-\lambda}\right]\\
   &=&    \frac{1-\lambda}{Z^{\lambda}_{\theta}}
  E_{\bar{p}_*}\left[\left(\frac{p_{\theta}(y|x)}{p_*(y|x)}\right)^{1-\lambda}\frac{\partial \log p_{\theta}(y|x)}{\partial
				 \theta}
				 \right]\\
   &=&    \frac{1-\lambda}{Z^{\lambda}_{\theta}}
  \int q_*(x)p_*(y|x)^{\lambda}p_{\theta}(y|x)^{1-\lambda}s_{\theta}(y|x)dxdy\\
   &=&    (1-\lambda)E_{\bar{p}^{\lambda}_{\theta}}[s_{\theta}(y|x)]. 
  \end{eqnarray*}
Therefore, the first derivative is 
\begin{eqnarray*}
 \frac{\partial d_{\lambda}(p_*,p_{\theta})}{\partial
  \theta}\!\!\!\!&=&\!\!\!\!-\frac{1}{1-\lambda}\frac{\partial \log Z^{\lambda}_{\theta}}{\partial \theta}=-E_{\bar{p}^{\lambda}_{\theta}}\left[s_{\theta}(y|x)\right].
\end{eqnarray*}
  Furthermore, we have
  \begin{eqnarray*}
   \frac{\partial \log \bar{p}^{\lambda}_{\theta}(x,y)}{\partial
    \theta}
    &=&
    \frac{\partial 
    }{\partial
    \theta}\log \left(\frac{q_*(x)p_*(y|x)^{\lambda}p_{\theta}(y|x)^{1-\lambda}}{Z^{\lambda}_{\theta}}\right)\\
    &=&
(1-\lambda)    \frac{\partial \log
    p_{\theta}(y|x)}{\partial
    \theta}-\frac{\partial \log Z^{\lambda}_{\theta}}{\partial
    \theta}\\
    &=&
(1-\lambda) s_{\theta}(y|x)-(1-\lambda)E_{\bar{p}^{\lambda}_{\theta}}[s_{\theta}(y|x)]\\
    &=&
(1-\lambda) \left(s_{\theta}(y|x)-E_{\bar{p}^{\lambda}_{\theta}}[s_{\theta}(y|x)]\right).
  \end{eqnarray*}
Hence, 
\begin{eqnarray*}
 \lefteqn{\frac{\partial^2d_{\lambda}(p_*,p_{\theta})}{\partial \theta\partial
  \theta^T}}\\
 &\hskip-4mm=&\hskip-3mm-
\int s_{\theta}(y|x)\bar{p}^{\lambda}_{\theta}(x,y)\left(\frac{\partial \log
				     \bar{p}^{\lambda}_{\theta}(x,y)}{\partial
				     \theta}\right)^T
\\
&&+\bar{p}^{\lambda}_{\theta}(x,y)\frac{\partial	s_{\theta}(y|x)}{\partial \theta^T}dxdy\\
 &\hskip-4mm=&\hskip-3mm-E_{\bar{p}^{\lambda}_{\theta}}\bigg[(1-\lambda)s_{\theta}(y|x)\left(s_{\theta}(y|x)-E_{\bar{p}^{\lambda}_{\theta}}[s_{\theta}(y|x)]\right)^T\\
&&  +\frac{\partial
				   ^2\log p_{\theta}(y|x)}{\partial \theta\partial\theta^T}\bigg]\\
&\hskip-4mm=&\hskip-3mm-E_{\bar{p}^{\lambda}_{\theta}}\left[\frac{\partial
				   ^2\log p_{\theta}(y|x)}{\partial
				   \theta\partial\theta^T}\right]-(1-\lambda)\\
&\hskip-4mm&\hskip-7mm\cdot
E_{\bar{p}^{\lambda}_{\theta}}\bigg[\!\!
\left(s_{\theta}(y|x)\!-\!E_{\bar{p}^{\lambda}_{\theta}}\left[s_{\theta}(y|x)\right]\right)\left(s_{\theta}(y|x)\!-\!E_{\bar{p}^{\lambda}_{\theta}}[s_{\theta}(y|x)]\right)^T\!\bigg]\\
&\hskip-4mm=&\hskip-3mm-E_{\bar{p}^{\lambda}_{\theta}}\left[\frac{\partial
				   ^2\log p_{\theta}(y|x)}{\partial
				   \theta\partial\theta^T}\right]\!\!-\!(1-\lambda)\mbox{Var}_{\bar{p}^{\lambda}_{\theta}}\left(s_{\theta}(y|x)\right).
\end{eqnarray*}
 \end{proof}

 \begin{lemma}\label{bqlemma}
  Let
\begin{eqnarray*}
&& \theta(\lambda):=\lambda\theta^*+(1-\lambda)\theta,\quad
  \bar{\theta}:=\theta-\theta^*,\quad  \bar{\theta}':=\Sigma^{1/2}\bar{\theta},
  \nonumber\\
&& c:=\frac{\sigma^2}{\lambda(1-\lambda)}.
\end{eqnarray*}
  If we assume that $p_*(y|x)=N(y|x^T\theta^*,\sigma^2)$ (\thatis,
  linear regression setting),
\begin{eqnarray}
 p^{\lambda}_{\theta}(y|x)\!\!\!\!&=&\!\!\!\!N(y|x^T\theta(\lambda),\sigma^2), \nonumber\\
 q^{\lambda}_{\theta}(y|x)\!\!\!\!&=&\!\!\!\!\frac{q_*(x)\exp\left(-\frac{1}{2c}(x^T\dtheta)^2\right)}{Z^{\lambda}_{\theta}}, \nonumber\\
 \frac{\partial d_{\lambda}(p_*,p_{\theta})}{\partial \theta}\!\!\!\!&=&\!\!\!\!
\frac{\lambda}{\sigma^2}E_{q^{\lambda}_{\theta}}[xx^T]\bar{\theta},\nonumber\\
 \frac{\partial^2d_{\lambda}(p_*,p_{\theta})}{\partial \theta\partial
  \theta^T}
  \!\!\!\!&=&\!\!\!\!
\frac{\lambda}{\sigma^2}E_{q^{\lambda}_{\theta}}[xx^T]-\frac{\lambda}{\sigma^2c}{\rm
Var}_{q^{\lambda}_{\theta}}\!\left(xx^T\bar{\theta}\right).\label{genhessianlasso}
\end{eqnarray}
If we additionally assume that $q_*(x)=N(x|\mib{0},\Sigma)$ with a
  non-singular covariance matrix $\Sigma$, 
\begin{eqnarray}
 q^{\lambda}_{\theta}(x)&=&N(x|\mib{0},\Sigma^{\lambda}_{\theta}),\nonumber\\
 \frac{\partial d_{\lambda}(p_*,p_{\theta})}{\partial \theta}\!\!\!\!&=&\!\!\!\!
\frac{\lambda}{\sigma^2}\left(\frac{c}{c+\|\bar{\theta}'\|^2_2}\right)\Sigma^{1/2}\bar{\theta}',\label{normalfirstderiv}\\
 \frac{\partial^2d_{\lambda}(p_*,p_{\theta})}{\partial \theta\partial
  \theta^T}
  \!\!\!\!&=&\!\!\!\!
   \frac{\lambda}{\sigma^2}\left(\frac{c}{c+\|\dtheta'\|_2^2}\right)\Sigma\nonumber\\
   &&
   -\frac{2\lambda}{\sigma^2}
\left(\frac{c}{(c+\|\dtheta'\|_2^2)^2}\right)\Sigma^{1/2}\dtheta'\left(\dtheta'\right)^T\Sigma^{1/2},\nonumber\\\label{normalhessianlasso}
\end{eqnarray}
  where
\[
  \esigma:=\Sigma-\frac{\Sigma^{1/2}\bar{\theta}'(\bar{\theta}')^T\Sigma^{1/2}}{c+\|\bar{\theta}'\|^2_2}. 
\]
\end{lemma}
\begin{proof}
By completing squares, we can rewrite
 $\bar{p}^{\lambda}_{\theta}(x,y)$ as
 \begin{eqnarray*}
  \lefteqn{\bar{p}^{\lambda}_{\theta}(x,y)}\\
  &&\hskip-7mm=\frac{q_*(x)}{\left(2\pi
				\sigma^2\right)^{\frac{1}{2}}Z^{\lambda}_{\theta}}\exp\left(-\frac{\left(\lambda(y-x^T\theta^*)^2+(1-\lambda)(y-x^T\theta)^2\right)}{2\sigma^2}\right)\\
  &&\hskip-7mm=
   \frac{q_*(x)}{\left(2\pi\sigma^2\right)^{\frac{n}{2}}Z^{\lambda}_{\theta}}\\
   &&\hskip-4mm\cdot \exp\left(-\frac{\left(
(y-x^T\theta(\lambda))^2+\lambda(1-\lambda)(x^T(\theta^*-\theta))^2
\right)}{2\sigma^2}\right)  \\
  &&\hskip-7mm=
   \frac{q_*(x)}{Z^{\lambda}_{\theta}}
   \exp\left(-\frac{
\lambda(1-\lambda)(x^T\dtheta)^2
}{2\sigma^2}\right)N(y|x^T\theta(\lambda),\sigma^2).
 \end{eqnarray*}
Hence, $p^{\lambda}_{\theta}(y|x)$ is
 $N(y|x^T\theta(\lambda),\sigma^2)$. Integrating $y$ out, we also have 
\begin{eqnarray*}
 q^{\lambda}_{\theta}(x)&=&
  \frac{q_*(x)\exp\left(-\frac{1}{2c}(x^T\bar{\theta})^2\right)}{Z^{\lambda}_{\theta}}.
\end{eqnarray*}
 When $q_*(x)=N(\mib{0},\Sigma)$, 
\begin{eqnarray}
 q^{\lambda}_{\theta}(x)
&=&\frac{\exp\left(-\frac{1}{2}x^T\Sigma^{-1}x-\frac{1}{2c}x^T\bar{\theta}\bar{\theta}^Tx\right)}{(2\pi)^{p/2}|\Sigma|^{1/2}Z^{\lambda}_{\theta}}\nonumber\\
&=&\frac{\exp\left(-\frac{1}{2}x^T\left(\Sigma^{-1}+\frac{1}{c}\bar{\theta}\bar{\theta}^T\right)x\right)}{(2\pi
 )^{p/2}|\Sigma|^{1/2}Z^{\lambda}_{\theta}}.\label{qcalc}
\end{eqnarray}
Since $\Sigma$ is strictly positive definite by the assumption,
 $\Sigma^{-1}+(1/c)\bar{\theta}\bar{\theta}^T$ is non-singular. 
Hence, by the inverse formula (Lemma \ref{inverseformula} in Appendix),
\begin{eqnarray}
\esigma&=&\left(\Sigma^{-1}+\frac{1}{c}\dtheta\dtheta^T\right)^{-1}=\Sigma-\frac{\Sigma\bar{\theta}\bar{\theta}^T\Sigma}{c+\bar{\theta}^T\Sigma\bar{\theta}}\nonumber \\
&=&\Sigma
 -\frac{\Sigma^{1/2}\bar{\theta}'(\bar{\theta}')^{T}\Sigma^{1/2}}{c+\|\bar{\theta}'\|^2_2}. \label{qsigma}
\end{eqnarray}
 Therefore, $q^{\lambda}_{\theta}(x)=N(x|\mib{0},\esigma)$.
The score function and Hessian of $\log p_{\theta}(y|x)$ are 
\begin{eqnarray}
 s_{\theta}(y|x)&=&\frac{1}{\sigma^2}x(y-x^T\theta),\nonumber\\
 \frac{\partial^2\log
  p_{\theta}(y|x)}{\partial \theta\partial \theta^T}&=&-\frac{1}{\sigma^2}xx^T.\label{secondderiv}
\end{eqnarray}
Using (\ref{firstderivative}), the first derivative is obtained as
\begin{eqnarray*}
\frac{\partial d_{\lambda}(p_*,p_{\theta})}{\partial
 \theta}&=&-E_{\bar{p}^{\lambda}_{\theta}}[s_{\theta}(y|x)]\\
 &=&-E_{q^{\lambda}_{\theta}}\left[E_{p^{\lambda}_{\theta}}[s_{\theta}(y|x)]\right]\\
&=&-E_{q^{\lambda}_{\theta}}\left[E_{p^{\lambda}_{\theta}}\left[\frac{1}{\sigma^2}x(y-x^T\theta)\right]\right]\\
 &=&-E_{q^{\lambda}_{\theta}}\left[\frac{1}{\sigma^2}xx^T(\theta(\lambda)-\theta)\right]\\
  &=&\frac{\lambda}{\sigma^2}E_{q^{\lambda}_{\theta}}\left[xx^T\right]\dtheta
\end{eqnarray*}
 because $\theta(\lambda)-\theta=-\lambda\bar{\theta}$.
 When $q_*(y|x)=N(\mib{0},\Sigma)$, 
 \[
  \frac{\partial d_{\lambda}(p_*,p_{\theta})}{\partial \theta}=\frac{\lambda}{\sigma^2}\esigma\bar{\theta}.
 \]
 From (\ref{qsigma}), we have
\begin{eqnarray}
 \esigma\dtheta&=&\Sigma\dtheta-\frac{\Sigma^{1/2}\bar{\theta}'(\bar{\theta}')^T\Sigma^{1/2}\dtheta}{c+\|\bar{\theta}'\|^2_2}\nonumber
  \\ 
 &=&\Sigma^{\frac{1}{2}}\dtheta'-\left(\frac{\|\dtheta'\|^2_2}{c+\|\bar{\theta}'\|^2_2}\right)\Sigma^{1/2}\bar{\theta}'\nonumber \\
 &=&\left(\frac{c}{c+\|\bar{\theta}'\|^2_2}\right)\Sigma^{1/2}\bar{\theta}', \label{Sthetabar}
\end{eqnarray}
 which gives (\ref{normalfirstderiv}). Though (\ref{normalhessianlasso})
 can be obtained by differentiating (\ref{normalfirstderiv}), we derive it
 by way of (\ref{hessiangeneral}) here.
To calculate the covariance matrix of $s_{\theta}$ in terms of
 $\bar{p}^{\lambda}_{\theta}$, we decompose $s_{\theta}$ as
\begin{eqnarray*}
 s_{\theta}(y|x)&=&\frac{1}{\sigma^2}x(y-x^T\theta(\lambda)+x^T\theta(\lambda)-x^T\theta)\\
 &=&\frac{1}{\sigma^2}x(y-x^T\theta(\lambda))-\frac{\lambda}{\sigma^2}xx^T\dtheta.
\end{eqnarray*}
 Note that the covariance of $(1/\sigma^2)x(y-x^T\theta(\lambda))$ and
 $-(\lambda/\sigma^2)xx^T\dtheta$ vanishes since
 \begin{eqnarray*}
  \lefteqn{E_{\bar{p}^{\lambda}_{\theta}}[x(y-x^T\theta(\lambda))(xx^T\dtheta)^T]}\\
  &=&E_{q^{\lambda}_{\theta}}\left[xx^T(x^T\dtheta)E_{p^{\lambda}_{\theta}}\left[(y-x^T\theta(\lambda))\right]\right]
   =0.
 \end{eqnarray*}
Therefore, we have
\begin{eqnarray*}
 \lefteqn{\mbox{Var}_{\bar{p}^{\lambda}_{\theta}}\left(s_{\theta}\right)}\\
 &=&
  \mbox{Var}_{\bar{p}^{\lambda}_{\theta}}\left(\frac{1}{\sigma^2}x(y-x^T\theta(\lambda))\right)
  +\mbox{Var}_{\bar{p}^{\lambda}_{\theta}}\left(\frac{\lambda}{\sigma^2}xx^T\dtheta\right)\\
 &=&
  \frac{1}{\sigma^4}E_{\bar{p}^{\lambda}_{\theta}}\left[(y-x^T\theta(\lambda))^2xx^T\right]
  +\frac{\lambda^2}{\sigma^4}\mbox{Var}_{q^{\lambda}_{\theta}}\left(xx^T\dtheta\right)\\
 &=&
  \frac{1}{\sigma^2}E_{q^{\lambda}_{\theta}}\left[xx^T\right]
  +\frac{\lambda^2}{\sigma^4}\mbox{Var}_{q^{\lambda}_{\theta}}\left(xx^T\dtheta\right)
\end{eqnarray*}
By (\ref{hessiangeneral}) combined with (\ref{secondderiv}), the Hessian of \renyi divergence is calculated as
\begin{eqnarray*}
 \lefteqn{\frac{\partial^2d_{\lambda}(p_*,p_{\theta})}{\partial \theta\partial
  \theta^T}}\\
 &=&\!\!\!\!\frac{1}{\sigma^2}E_{\bar{p}^{\lambda}_{\theta}}[xx^T]-(1-\lambda)\left(\frac{1}{\sigma^2}E_{q^{\lambda}_{\theta}}[xx^T]+\frac{\lambda^2}{\sigma^4}\mbox{Var}_{q^{\lambda}_{\theta}}\left(xx^T\bar{\theta}\right)\right)\\
&=&\!\!\!\!\frac{\lambda}{\sigma^2}E_{q^{\lambda}_{\theta}}[xx^T]-\frac{\lambda^2(1-\lambda)}{\sigma^4}\mbox{Var}_{q^{\lambda}_{\theta}}\left(xx^T\bar{\theta}\right) \\
&=&\!\!\!\!\frac{\lambda}{\sigma^2}E_{q^{\lambda}_{\theta}}[xx^T]-\frac{\lambda}{\sigma^2c}\mbox{Var}_{q^{\lambda}_{\theta}}\left(xx^T\bar{\theta}\right). 
\end{eqnarray*}
 When $q_*(x)=N(\mib{0},\Sigma)$,
 $\mbox{Var}_{q^{\lambda}_{\theta}}\left(xx^T\bar{\theta}\right)$ is
 calculated as follows.
Note that 
\begin{eqnarray*}
 \mbox{Var}_{q^{\lambda}_{\theta}}(xx^T\dtheta)
=E_{q^{\lambda}_{\theta}}\left[(xx^T\dtheta)(xx^T\dtheta)^T\right]
-(\esigma\dtheta)(\esigma\dtheta)^T.
\end{eqnarray*}
The $(j_1,j_2)$ element of
$E_{q^{\lambda}_{\theta}}\left[xx^T\dtheta\dtheta^Txx^T\right]$
is calculated as
\begin{eqnarray*}
 E_{q^{\lambda}_{\theta}}\left[\left(xx^T\dtheta\dtheta^Txx^T\right)_{j_1j_2}\right]
\!=\!\!\!\!\sum_{j_3,j_4=1}^p\!\!\!\dtheta_{j_3}\dtheta_{j_4}E_{q^{\lambda}_{\theta}}\left[x_{j_1}x_{j_2}x_{j_3}x_{j_4}\right],
\end{eqnarray*}
where $x_j$ denotes the $j$th element of $x$ only here. 
Thus, we need all the fourth-moments of
$q_{\theta}^{\lambda}(x)$. We rewrite $\esigma$ as $S$ to reduce
 notation complexity hereafter. By the formula of moments of Gaussian distribution, we have
\[
 E_{q^{\lambda}_{\theta}}\left[x_{j_1}x_{j_2}x_{j_3}x_{j_4}\right]=S_{j_1j_2}S_{j_3j_4}+S_{j_1j_3}S_{j_2j_4}+S_{j_2j_3}S_{j_1j_4}.
\] 
Therefore, the above quantity is calculated as
\begin{eqnarray*}
\lefteqn{E_{q^{\lambda}_{\theta}}\left[\left(xx^T\dtheta\dtheta^Txx^T\right)_{j_1j_2}\right]}\\
&=&\!\!\!\!\sum_{j_3,j_4=1}^p\dtheta_{j_3}\dtheta_{j_4}(S_{j_1j_2}S_{j_3j_4}+S_{j_1j_3}S_{j_2j_4}+S_{j_2j_3}S_{j_1j_4})\\
&=&\dtheta^TS\dtheta S_{j_1j_2}+2(S\dtheta)_{j_1}(S\dtheta)_{j_2}.
\end{eqnarray*}
Summarizing these as a matrix form, we have
\begin{eqnarray*}
 E_{q^{\lambda}_{\theta}}\left[xx^T\dtheta\dtheta^Txx^T\right]
&=&(\dtheta^TS\dtheta)S+2S\dtheta(S\dtheta)^T.
\end{eqnarray*}
As a result,
$\mbox{Var}_{q^{\lambda}_{\theta}}(xx^T\dtheta)$
is obtained as
\begin{eqnarray}
 \mbox{Var}_{q^{\lambda}_{\theta}}(xx^T\dtheta)
&=&(\dtheta^TS\dtheta)S+2S\dtheta\dtheta^TS
-S\dtheta\dtheta^TS\nonumber\\
&=&S\dtheta\dtheta^TS+(\dtheta^TS\dtheta)S. \label{varst}
\end{eqnarray}
 Using (\ref{Sthetabar}), the first and second terms of (\ref{varst}) are calculated as
\begin{eqnarray*}
 S\dtheta\dtheta^TS
  &=&\left(\frac{c^2}{(c+\|\dtheta'\|_2^2)^2}\right)\Sigma^{1/2}\dtheta'\left(\dtheta'\right)^T\Sigma^{1/2},\\
  \dtheta^TS\dtheta
   &=&\left(\frac{c}{c+\|\bar{\theta}'\|^2_2}\right)(\dtheta)^T\Sigma^{1/2}\bar{\theta}'\\
   &=&\frac{c\|\dtheta'\|_2^2}{c+\|\dtheta'\|_2^2}.
 \end{eqnarray*}
 Combining these,
 \begin{eqnarray*}
\lefteqn{\frac{\partial^2d_{\lambda}(p_*,p_{\theta})}{\partial \theta\partial
 \theta^T}}\\
  &=&
\frac{\lambda}{\sigma^2}S-\frac{\lambda}{\sigma^2c}
   \bigg(
\left(\frac{c^2}{(c+\|\dtheta'\|_2^2)^2}\right)\Sigma^{1/2}\dtheta'\left(\dtheta'\right)^T\Sigma^{1/2}\\
  &&+
   \left(\frac{c\|\dtheta'\|_2^2}{c+\|\dtheta'\|_2^2}\right)S
   \bigg)\\
  &=&
   \frac{\lambda}{\sigma^2}\left(\frac{c}{c+\|\dtheta'\|_2^2}\right)S\\
   &&
   -\frac{\lambda}{\sigma^2}
\left(\frac{c}{(c+\|\dtheta'\|_2^2)^2}\right)\Sigma^{1/2}\dtheta'\left(\dtheta'\right)^T\Sigma^{1/2}
\\
  &=&
   \frac{\lambda}{\sigma^2}\left(\frac{c}{c+\|\dtheta'\|_2^2}\right)\left(\Sigma-\frac{\Sigma^{1/2}\dtheta'(\dtheta')^T\Sigma^{1/2}}{c+\|\dtheta'\|^2_2}\right)\\
   &&
   -\frac{\lambda}{\sigma^2}
\left(\frac{c}{(c+\|\dtheta'\|_2^2)^2}\right)\Sigma^{1/2}\dtheta'\left(\dtheta'\right)^T\Sigma^{1/2}\\
  &=&
   \frac{\lambda}{\sigma^2}\left(\frac{c}{c+\|\dtheta'\|_2^2}\right)\Sigma\\
   &&
   -\frac{2\lambda}{\sigma^2}
\left(\frac{c}{(c+\|\dtheta'\|_2^2)^2}\right)\Sigma^{1/2}\dtheta'\left(\dtheta'\right)^T\Sigma^{1/2}.
 \end{eqnarray*}
\end{proof}
\subsection{Upper Bound of Negative Hessian}\label{hessianlemmasec}
Using Lemma \ref{bqlemma} in Section \ref{bqlemmasec}, we show that the negative Hessian of the
\renyi divergence is bounded from above. 
  \begin{lemma}\label{hessianlemma}
Assume that $q_*(x)=N(x|\mib{0},\Sigma)$ and
   $p_*(y|x)=N(y|x^T\theta^*,\sigma^2)$, where $\Sigma$ is non-singular.
   For any $\theta,\theta^*$,
\begin{eqnarray}
 -\frac{\partial^2d_{\lambda}(p_*,p_{\theta})}{\partial \theta\partial
  \theta^T}
  \preceq
\frac{\lambda}{8\sigma^2}\Sigma, \label{hessianinequality}
\end{eqnarray}
  where $A\preceq B$ implies that $B-A$ is positive semi-definite. 
 \end{lemma}
 \begin{proof}
  By Lemma \ref{bqlemma}, we have
  \begin{eqnarray*}
 -\frac{\partial^2d_{\lambda}(p_*,p_{\theta})}{\partial \theta\partial
  \theta^T}
  \!\!\!\!&=&\!\!\!\!
   \frac{2\lambda}{\sigma^2}
\left(\frac{c}{(c+\|\dtheta'\|_2^2)^2}\right)\Sigma^{1/2}\dtheta'\left(\dtheta'\right)^T\Sigma^{1/2}\\
   &&
   -\frac{\lambda}{\sigma^2}\left(\frac{c}{c+\|\dtheta'\|_2^2}\right)\Sigma.
  \end{eqnarray*}
 For any nonzero vector $v\in \Re^p$,
 \begin{eqnarray*}
  v^T\Sigma^{1/2}\dtheta'\left(\dtheta'\right)^T\Sigma^{1/2}v
   &=& \left(v^T\Sigma^{1/2}\dtheta'\right)^2\\
   &\le& \|\Sigma^{1/2}v\|_2^2\cdot \|\dtheta'\|^2_2=v^T(\|\dtheta'\|^2_2\,\Sigma) v
 \end{eqnarray*}
 by Cauchy-Schwartz inequality. Hence, we have
 \[
  \Sigma^{1/2}\dtheta'\left(\dtheta'\right)^T\Sigma^{1/2}\preceq \|\dtheta'\|^2_2\,\Sigma.
 \]
 Thus,
 \begin{eqnarray*}
\lefteqn{-\frac{\partial^2d_{\lambda}(p_*,p_{\theta})}{\partial \theta\partial
 \theta^T}}\\
 &\preceq&
   \frac{2\lambda}{\sigma^2}
\left(\frac{c\|\dtheta'\|_2^2}{(c+\|\dtheta'\|_2^2)^2}\right)\Sigma
-\frac{\lambda}{\sigma^2}\left(\frac{c}{c+\|\dtheta'\|_2^2}\right)\Sigma\\
  &=&
   \frac{\lambda}{\sigma^2}
\left(\frac{c(\|\dtheta'\|_2^2-c)}{(c+\|\dtheta'\|_2^2)^2}\right)\Sigma.
 \end{eqnarray*}
 Define
 \[
 f(t):=\frac{c(t-c)}{(c+t)^2}
 \]
 for $t\ge 0$. Checking the properties of $f(t)$, we have
 \begin{eqnarray*}
  f(0)&=&-1,\\
  f(c)&=&0,\\
  f(\infty)&=&0,\\
   \frac{df(t)}{dt}&=&\frac{c(3c-t)}{(t+c)^3}.
 \end{eqnarray*}
 Therefore, $\max_{t\in [0,\infty)}f(t)=f(3c)=1/8$. As a result, we
 obtain
 \[
-\frac{\partial^2d_{\lambda}(p_*,p_{\theta})}{\partial \theta\partial
 \theta^T}\preceq
   \frac{\lambda}{8\sigma^2}\Sigma.
 \]
 \end{proof}

\subsection{Proof of Lemma \ref{erv4lasso}}\label{erv4lassoproof}
We are now ready to derive $\epsilon$-risk valid weighted $\ell_1$
penalties. 
 \begin{proof}
Similarly to the rewriting from (\ref{rv}) to (\ref{rv2}), we can
  rewrite the condition for $\epsilon$-risk validity as 
\begin{eqnarray}
\lefteqn{\forall x^n \in \typicalset, \forall y^n\in \rs{Y}^n, \forall
\theta\in \Theta,} \nonumber\\  
&&\hskip-7mm\min_{\ctheta \in \cTheta(q_*)}
 \Bigl\{
 \underbrace{
d^n_{\lambda}(p_*,p_{\theta})-
d^n_{\lambda}(p_*,p_{\ctheta})
}_{\mbox{loss variation part}}
+
 \underbrace{
\log
\frac{p_{\theta}(y^n|x^n)}{p_{\ctheta}(y^n|x^n)}+\ctpen(\ctheta|q_*)
}_{\mbox{codelength validity part}}
\Bigr\}
\nonumber \\
&& \hskip-7mm\le \copen(\theta|x^n).\label{erv2}
\end{eqnarray}
  We again write the inside part of the minimum in (\ref{erv2}) as
  $H(\theta,\ctheta,x^n,y^n)$.
As described in Section \ref{bctheory}, the direct minimization of
 $H(\theta,\ctheta,x^n,y^n)$ seems to be difficult.
Instead of evaluating the minimum explicitly, we borrow a nice randomization
technique introduced in \cite{chatterjeebarron14} with some
 modifications.
Their key idea is to evaluate not $\min_{\ctheta
 }H(\theta,\ctheta,x^n,y^n)$ directly but its expectation $E_{\ctheta
 }[H(\theta,\ctheta,x^n,y^n)]$ with respect 
to a dexterously randomized $\ctheta $ because the expectation is 
larger than the minimum.  
Let us define $w^*:=(w_1^*,w_2^*,\cdots, w_p^*)^T$, where
$w_j^*=\sqrt{\Sigma_{jj}}$ and $W^*:=\mbox{diag}(w^*_1,\cdots, w^*_p)$. 
  We quantize $\Theta$ as
\begin{equation}
  \cTheta(q_*):=\{\delta (W^*)^{-1}z|z\in
 \rs{Z}^p\}, \label{quantizationrd}
\end{equation}
 where $\delta>0$ is a quantization width and $\rs{Z}$ is the set of all
 integers.
 Though $\cTheta$ depends on $x^n$ in fixed design cases
 \cite{chatterjeebarron14}, we must remove the dependency to satisfy the
 $\epsilon$-risk validity as above. 
For each $\theta$, $\ctheta $ is randomized as 
\begin{eqnarray}
 \ctheta _{j}=\left\{
\begin{array}{cl}
\frac{\delta}{w_j^*} \lceil m_{j}\rceil& \mbox{ with prob. } m_{j}-\lfloor m_{j}\rfloor\\
\frac{\delta}{w_j^*} \lfloor m_{j}\rfloor & \mbox{ with prob. }\lceil
 m_{j}\rceil-m_{j}\\
\frac{\delta}{w_j^*} m_{j} & \mbox{ with prob. }1-(\lceil m_{j}\rceil
 -\lfloor m_{j}\rfloor )\\
\end{array}
\right., \label{tildethetarandom}
\end{eqnarray}
where $m_j:=w_j^*\theta_j/\delta$ and each component of $\ctheta $
 is statistically independent of each other. Its important properties are
\begin{eqnarray}
 &&E_{\ctheta}[\ctheta]=\theta,\quad \mbox{(unbiasedness)}\nonumber\\
 &&E_{\ctheta}[|\ctheta|]=|\theta|, \label{absoluteunbiasedness}\\
 &&E_{\ctheta}[(\ctheta_j-\theta_j)(\ctheta_{j'}-\theta_{j'})]\le
I(j=j')\frac{\delta}{w_j^*}|\theta_j|, \nonumber 
\end{eqnarray}
 where $|\ctheta|$ denotes a vector whose $j$th component is the absolute
 value of $\ctheta_j$ and similarly for $|\theta|$. 
Using these, we can bound $E_{\ctheta}[H(\theta,\ctheta,x^n,y^n)]$ as follows.
The loss variation part in (\ref{erv2}) is the main concern because it
is more complicated than squared error of fixed design cases. Let us
consider the following Taylor expansion 
\begin{eqnarray}
\lefteqn{\hskip-0.8cm d^n_{\lambda}(p_*,p_{\theta})-d^n_{\lambda}(p_*,p_{\ctheta})=-\left(\frac{\partial
 d^n_{\lambda}(p_*,p_{\theta})}{\partial\theta}\right)^T(\ctheta-\theta)}\nonumber\\
 &&-\frac{1}{2}\trace\left(\frac{\partial^2d^n_{\lambda}(p_*,p_{\theta^{\circ}})}{\partial\theta\partial
											     \theta^T}(\ctheta -\theta)(\ctheta -\theta)^T\right), \label{type1expansion}
\end{eqnarray}
where $\theta^{\circ}$ is a vector between $\theta$ and $\ctheta $. The first term in the right side of
\eqnumn{type1expansion} vanishes after taking expectation
with respect to $\ctheta $ because
$E_{\ctheta }[\ctheta -\theta]=0$.
 As for the second term, we obtain 
\begin{eqnarray*}
&&\trace\left(\!-\frac{\partial^2d^n_{\lambda}(p_*,p_{\theta^{\circ}})}{\partial\theta\partial
		  \theta^T}(\ctheta -\theta)(\ctheta -\theta)^T\right)\\
&&\le
\frac{n\lambda}{8\sigma^2}\trace\left(\Sigma\left(\ctheta -\theta\right)\left(\ctheta -\theta\right)^T\right)
\end{eqnarray*}
 by Lemma \ref{hessianlemma}. Thus, expectation of the loss variation part with respect to $\ctheta$ is
  bounded as
\begin{equation}
E_{\ctheta}\left[d^n_{\lambda}(p_*,p_{\theta})-d^n_{\lambda}(p_*,p_{\ctheta })\right]\le
\frac{\delta n\lambda}{16\sigma^2}\|\theta\|_{w^*,1}. \label{elvp}
\end{equation}
The codelength validity part in (\ref{erv2}) have the same form as
that for the fixed design case in its appearance. However, we need to evaluate it
again in our setting because both $\cTheta$ and $\ctpen$ are different
 from those in \cite{chatterjeebarron14}. 
The likelihood term is calculated as
\begin{eqnarray*}
\frac{1}{2\sigma^2}\Big(\!2(Y-X\theta)^TX(\theta-\ctheta )
\!\!+\!\!\trace\big(X^TX(\ctheta -\theta)(\ctheta -\theta)^T\big)\!\Big).
\end{eqnarray*}
  Taking expectation with respect to $\ctheta $, we have
\begin{eqnarray}
\hskip-6mm E_{\ctheta }\left[\log\frac{p_{\theta}(y^n|x^n)}{p_{\ctheta }(y^n|x^n)}\right]
\!\!\!\!\!&=&\!\!\!\!\! \frac{n}{2\sigma^2}E_{\ctheta }\!\left[\trace\!\left(W^2(\ctheta -\theta)(\ctheta -\theta)^T\right)\right]\nonumber\\
\!\!\!\!\!&\le&\!\!\!\!\!\frac{\delta n}{2\sigma^2}\sum_{j=1}^p\frac{w_j^2}{w_j^*}|\theta_j|,\nonumber
\end{eqnarray}
where $W:= \mbox{diag}(w_1,w_2,\cdots, w_p)$.
We define a codelength
function $C(z):=\|z\|_1\log 4p+\log 2$ over $\rs{Z}^p$.  Note that $C(z)$ satisfies Kraft's inequality.
  Let us define a codelength function on $\cTheta(q_*)$ as
\begin{eqnarray}
 \ctpen(\ctheta |q_*):=\frac{1}{\beta}C\left(\frac{1}{\delta}W^*\ctheta
				       \right)=\frac{1}{\beta\delta}\|W^*\ctheta
 \|_1\log 4p+\frac{\log 2}{\beta}. \label{ctpen4lasso}
\end{eqnarray}
By this definition, $\ctpen$ satisfies $\beta$-stronger Kraft's inequality and
 does not depend on $x^n$ but depends on $q_*(x)$ through $W^*$.
By taking expectation with respect to $\ctheta $, we have
\begin{eqnarray*}
 E_{\ctheta }\left[\ctpen(\ctheta |q_*)\right]=\frac{\log 4p}{\beta\delta}\|\theta\|_{w^*,1}+\frac{\log 2}{\beta}
\end{eqnarray*}
 because of (\ref{absoluteunbiasedness}). 
Thus, the codelength validity part is bounded above by
\begin{equation}
 \frac{\delta n}{2\sigma^2}\sum_{j=1}^p\frac{w_j^2}{w_j^*}|\theta_j|+\frac{\log 4p}{\beta\delta}\|\theta\|_{w^*,1}+\frac{\log 2}{\beta}.\nonumber
\end{equation}
Combining with the loss variation part, we obtain an upper bound of
$E_{\ctheta }[H(\theta,\ctheta,x^n,y^n)]$ as
\[
\frac{\delta n\lambda}{16\sigma^2}
\|\theta\|_{w^*,1}
\!+\!\frac{\delta n}{2\sigma^2}\sum_{j=1}^p\frac{w_j^{2}}{w_j^*}|\theta_j|
\!+\!\frac{\log 4p}{\beta\delta}\|\theta\|_{w^*,1}
\!+\!\frac{\log
2}{\beta}. 
\]
 Since $x^n\in \typicalset$, we have
 \[
 \sqrt{(1-\epsilon)}w_j^*\le w_j\le \sqrt{(1+\epsilon)}w_j^*.
 \]
Thus, we can bound $E_{\ctheta }[H(\theta,\ctheta,x^n,y^n)]$ by the data-dependent weighted $\ell_1$ norm $\|\theta\|_{w,1}$ as 
\begin{eqnarray}
 \lefteqn{E_{\ctheta }[H(\theta,\ctheta,x^n,y^n)]}\nonumber\\
&\hskip-3mm\le& \hskip-2mm\frac{\delta n\lambda}{16\sigma^2}
\frac{\|\theta\|_{w,1}}{\sqrt{1-\epsilon}}
+\frac{\delta n\sqrt{1+\epsilon}}{2\sigma^2}\sum_{j=1}^p\frac{w_j^{2}}{w_j}|\theta_j|
+\frac{\log 4p}{\beta\delta}\frac{\|\theta\|_{w,1}}{\sqrt{1-\epsilon}}\nonumber\\
&&+\frac{\log 2}{\beta}\nonumber\\
&\hskip-3mm=&\hskip-2mm\left(
  \frac{\delta n}{\sigma^2}
  \left(\frac{\lambda}{16\sqrt{1-\epsilon}}\!+\!\frac{\sqrt{1+\epsilon}}{2}\right)\!+\!
  \frac{\log 4p}{\delta \beta\sqrt{1-\epsilon}}\right)
\|\theta\|_{w,1}\!+\!\frac{\log 2}{\beta}.\nonumber
\end{eqnarray}
Because this holds for any $\delta>0$, we can minimize
the upper bound with respect to $\delta $, which completes the proof.
 \end{proof}
\subsection{Some Remarks on the Proof of Lemma \ref{erv4lasso}}\label{someremarks}
The main difference of the proof from the fixed design case is in the loss
variation part. In the fixed design case, the \renyi divergence
$d_{\lambda}(p_*,p_{\theta}|x^n)$ is convex in terms of $\theta$. When the
\renyi divergence is convex, the negative Hessian is negative
semi-definite for all $\theta$. Hence, the loss variation part is
trivially bounded above by zero. On the other hand,
$d_{\lambda}(p_*,p_{\theta})$ is not convex in terms of $\theta$.
This can be intuitively seen by deriving the explicit form of
$d_{\lambda}(p_*,p_{\theta})$ instead of checking the positive
semi-definiteness of its Hessian.
From (\ref{qcalc}), we have
\begin{eqnarray}
\hskip-9mm Z^{\lambda}_{\theta}
  &=&
\int \frac{\exp\left(-\frac{1}{2}\left(x^T(\esigma)^{-1}x\right)\right)}{(2\pi)^{p/2}|\Sigma|^{1/2}}dx\nonumber\\
&=&|\Sigma|^{-1/2}|\esigma|^{1/2}=|\Sigma^{-1/2}\esigma\Sigma^{-1/2}|^{1/2}\nonumber\\
&=& 
 \left|I_p-\left(\frac{1}{c+\|\bar{\theta}'\|_2^2}\right)\bar{\theta}'\left(\bar{\theta}'\right)^T\right|^{1/2}\nonumber\\
 &=&
  \left|I_p-\left(\frac{\|\bar{\theta}'\|_2^2}{c+\|\bar{\theta}'\|_2^2}\right)\left(\frac{\bar{\theta}'}{\|\bar{\theta}'\|_2}\right)\left(\frac{\bar{\theta}'}{\|\bar{\theta}'\|_2}\right)^T\right|^{1/2}\!\!,\label{oldrenyi}
\end{eqnarray}
where $I_p$ is the identity matrix of dimension $p$. 
Prof. A. R. Barron suggested in a 
private discussion that $Z^{\lambda}_{\theta}$ can be
simplified more as follows. 
Let $Q:=[q_1,q_2,\cdots, q_p]$ be an orthogonal matrix such that
$q_1:=\bar{\theta}'/\|\bar{\theta}'\|_2$. Using this, we have
\begin{eqnarray*}
 \lefteqn{I_p-\left(\frac{\|\bar{\theta}'\|_2^2}{c+\|\bar{\theta}'\|_2^2}\right)\left(\frac{\bar{\theta}'}{\|\bar{\theta}'\|_2}\right)\left(\frac{\bar{\theta}'}{\|\bar{\theta}'\|_2}\right)^T}\\
 &=&\!\!\!\!QQ^T\!-\!\left(\frac{\|\bar{\theta}'\|_2^2}{c+\|\bar{\theta}'\|_2^2}\right)q_1q_1^T\\
 &=&\left(1-\left(\frac{\|\bar{\theta}'\|_2^2}{c+\|\bar{\theta}'\|_2^2}\right)\right)q_1q_1^T\!+\!\sum_{j=2}^pq_jq_j^T\\
 &=&\left(\frac{c}{c+\|\bar{\theta}'\|_2^2}\right)q_1q_1^T\!+\!\sum_{j=2}^pq_jq_j^T\\
 &=&\!\!\!\!Q\left(
     \begin{array}{ccccc}
      c/(c+\|\dtheta'\|^2_2) & 0 & 0 & \cdots &  0\\
      0 &1 & 0 & \cdots & 0\\
      0 &0 & 1 & \cdots &0\\
      \vdots & \vdots & \vdots & \vdots & \vdots\\
      0 &0 & 0&\cdots  & 1\\
     \end{array}
    \right)Q^T.
\end{eqnarray*}
Hence, the resultant $Z^{\lambda}_{\theta}$ is obtained as
\begin{eqnarray*}
 Z^{\lambda}_{\ctheta}&=&\left|I_p-\gamma(\|\bar{\theta}'\|_2^2)\left(\frac{\bar{\theta}'}{\|\bar{\theta}'\|_2}\right)\left(\frac{\bar{\theta}'}{\|\bar{\theta}'\|_2}\right)^T\right|^{\frac{1}{2}}\\
  &=&\left(\frac{c}{c+\|\bar{\theta}'\|_2^2}\right)^{\frac{1}{2}}.
\end{eqnarray*}
Thus, we have a simple expression of the \renyi divergence as
\begin{equation}
 d_{\lambda}(p_*,p_{\theta})=\frac{1}{2(1-\lambda)}\log \left(1+\frac{\|\dtheta'\|^2_2}{c}\right).\label{andrewsuggestion}
\end{equation}
From this form, we can easily know that the \renyi divergence is not
convex. When the \renyi divergence is non-convex, it is unclear in
general whether and how the loss variation part is bounded above. This is one
of the main reasons why the derivation becomes more difficult than that
of the fixed design case. 

We also mention an alternative proof of Lemma \ref{erv4lasso} based on (\ref{andrewsuggestion}). We
provided Lemma \ref{renyiview} to calculate Hessian of the \renyi
divergence. 
However, the above simple expression of the
\renyi divergence is somewhat easier to differentiate, while the
expression based on (\ref{oldrenyi}) is somewhat hard to do it. Therefore, we can twice differentiate the above \renyi
divergence directly in order to obtain Hessian instead of Lemma
\ref{bqlemma} in our Gaussian setting. However, there is no guarantee
that such a simplification is always possible in general setting. 
In our proof, we tried to give a somewhat systematic way which is
easily applicable to other settings to some extent. 
Suppose now, for example, we are aim at deriving $\epsilon$-risk valid $\ell_1$
penalties for lasso when $q_*(x)$ is subject to non-Gaussian
distribution. 
By (\ref{genhessianlasso}) in Lemma \ref{bqlemma}, it suffices only to
bound $\mbox{Var}_{q^{\lambda}_{\theta}}(xx^T\dtheta)$ in the sense of
positive semi-definiteness because $-E_{q^{\lambda}_{\theta}}[xx^T]$ is
negative semi-definite.   
In general, it seemingly depends on a situation which is
better, the direct differential or using (\ref{genhessianlasso}).
In our Gaussian setting, we imagine that the easiest way to calculate
Hessian for most readers is to calculate the first derivative by the
formula (\ref{firstderivative}) and then to differentiate it directly,
though this depends on readers' background knowledge. 
For other settings, we believe that providing Lemmas \ref{renyiview} and
\ref{bqlemma} would be useful in some cases. 
\subsection{Proof of Lemma \ref{expbound}}\label{expboundproof}
Here, we show that $x^n$ distributes out of $\typicalset$ with
exponentially small probability with respect to $n$. 
 \begin{proof}
  The typical set $\typicalset$ can be decomposed covariate-wise as
\begin{eqnarray*}
 \typicalset &=&\Pi_{j=1}^p\typicalset(j),\\
 \typicalset(j)
  &:=&\left\{\mib{x}_j\in \Re^n\,\big|
			\left|(w_j^*)^2-(\|\mib{x}_j\|^2_2/n)\right|\le
			\epsilon (w_j^*)^2\right\}\\
  &=&\left\{\mib{x}_j\in \Re^n\,\big|
			\left|(w_j^*)^2-w_j^2)\right|\le
			\epsilon (w_j^*)^2\right\}, 
\end{eqnarray*}
where $\mib{x}_j:=(x_{1j},x_{2j},\cdots, x_{nj})^T$ and the above $\Pi$
  denotes a direct product of sets. From its definition, $w_j^2$ is subject to a Gamma distribution
  ${\rm Ga}((n/2),(2s)/n)$ when $\mib{x}_j\sim
  \Pi_{i=1}^nN(x_j|0,(w_j^*)^2)$.
 We write $w_j^2$ as $z$ and $(w_j^*)^2$ as $s$ (the index $j$ is
  dropped for legibility). We rewrite the Gamma distribution $g(z;s)$ in the form of exponential family:
\begin{eqnarray*}
 g(z;s) \!\!\!\!&:=&\!\!\!\!{\rm Ga}\left(\frac{n}{2},\frac{2s}{n}\right)
=\frac{\Gamma(\frac{n}{2})}z^{\frac{n}{2}-1}\exp\left(-\frac{nz}{2s}\right)}{\left(\frac{2s}{n}\right)^{\frac{n}{2}}\\
&=&\!\!\!\!\exp\bigg(\frac{n-2}{2}\log
 z-\frac{nz}{2s}-\log \left(\frac{2s}{n}\right)^{\frac{n}{2}}\Gamma\left(\frac{n}{2}\right)\bigg)\\
&=&\!\!\!\!\exp\left(C(z)+\nu z-\psi(\nu)\right),
\end{eqnarray*}
 where
 \begin{eqnarray*}
  C(z)&:=&\left(\frac{n-2}{2}\right)\log z,\quad 
   \nu:=-\frac{n}{2s},\\
  \psi(\nu)&:=&\log (-\nu)^{-n/2}\Gamma(n/2).
 \end{eqnarray*}
That is, $\nu$ is a natural parameter and $z$ is a sufficient statistic,
 so that the expectation parameter $\eta(s)$ is $E_{g(z;s)}[z]$. The
  relationship between the variance parameter $s$ and
  natural/expectation parameters are summarized as 
\[
 \nu(s):=-\frac{n}{2s},\quad\eta(\nu)=-\frac{n}{2\nu}.
\]
For exponential families, there is a useful Sanov-type inequality (Lemma
  \ref{exp2bound} in Appendix).
Using this Lemma, we can bound $\Pr(\mib{x}_j\notin
\typicalset(j))$ as follows. 
For this purpose, it suffices to bound the probability of the event
  $|w_j^2-w_j^{*2}|\le w_j^{*2}\epsilon$. When $s=(w_j^*)^2$ and $s'=s(1\pm
\epsilon)$, 
\begin{eqnarray*}
 \lefteqn{\rs{D}(\nu(s\pm
  \epsilon s),\nu)}\\
&=&\left(-\frac{n}{2s(1\pm \epsilon)}-\left(-\frac{n}{2s}\right)\right)s(1\pm
  \epsilon)-\frac{n}{2}\log(1\pm \epsilon)\\
&=&\left(-\frac{n}{2s}\right)\left(\frac{1}{(1\pm \epsilon)}-1\right)s(1\pm
  \epsilon)-\frac{n}{2}\log(1\pm \epsilon)\\
&=&\left(-\frac{n}{2}\right)\left(1-(1\pm \epsilon)\right)-\frac{n}{2}\log(1\pm \epsilon)\\
&=&\frac{n}{2}\left(\pm \epsilon-\log(1\pm \epsilon)\right),
\end{eqnarray*}
where $\rs{D}$ is the single data version of the KL-divergence defined
  by (\ref{kldiv}). 
It is easy to see that $\epsilon-\log(1+\epsilon)\le
-\epsilon-\log(1-\epsilon)$ for any $0<\epsilon<1$. By Lemma
  \ref{exp2bound}, we obtain
\begin{eqnarray*}
 \lefteqn{\Pr(|w_j^2-w_j^{*2}|\le \epsilon w_j^{*2})}\\
&=& 1-\Pr(w_j^2-w_J^{*2}\ge \epsilon w_j^{*2}\mbox{ or }w_J^{*2}-w_j^2\ge \epsilon w_j^{*2})\\
&=& 1-\Pr(w_j^2-w_J^{*2}\ge \epsilon w_j^{*2})-\Pr(w_J^{*2}-w_j^2\ge \epsilon w_j^{*2})\\
 &\ge& 1-\exp\left(-\frac{n}{2}(\epsilon-\log(1+\epsilon))\right)\\
 &&-\exp\left(-\frac{n}{2}(-\epsilon-\log(1-\epsilon))\right)\\
&\ge& 1-2\exp\left(-\frac{n}{2}(\epsilon-\log(1+\epsilon))\right).
\end{eqnarray*}
Hence $\inprob$ can be bounded below as
\begin{eqnarray*}
 \inprob&=&\Pr(x^n\in \typicalset)=\Pi_{j=1}^p(1-\Pr(\mib{x}_j\notin
  \typicalset(j)))\\
&\ge &\left(1-2\exp\left(-\frac{n}{2}(\epsilon-\log(1+\epsilon))\right)\right)^p\\
&\ge &1-2p\exp\left(-\frac{n}{2}(\epsilon-\log(1+\epsilon))\right).
\end{eqnarray*}
The last inequality follows from $(1-t)^p\ge 1-pt$ for any $t\in
  [0,1]$ and $p\ge 1$.
To simplify the bound, we can do more. 
The maximum positive real number $a$
  such that, for any $\epsilon\in [0,1]$, $a\epsilon^2\le
  (1/2)(\epsilon-\log(1+\epsilon))$ is $(1-\log 2)/2$. Then, the
  maximum integer $a_1$ such that $(1-\log 2)/2\ge 1/a_1$ is $7$, which
  gives the last inequality in the statement.
 \end{proof}
 \subsection{Proof of Lemma \ref{cv4lasso}}\label{proofcv4lasso}
We can prove this lemma by checking the proof of Lemma \ref{erv4lasso}.
  \begin{proof}
Let
     \[
      L_1(\theta|x^n):=\mu_1\|\theta\|_{w,1}+\mu_2.
     \]
      Similarly to the rewriting from (\ref{rvrd}) to (\ref{rv2}),
      we can restate the codelength validity condition for
      $L_1(\theta|x^n)$ as ``there exist a quantize subset $\cTheta(x^n)$
      and a model description length $\ctpen(\ctheta|x^n)$ satisfying
      the usual Kraft's inequality, such that
\begin{eqnarray}
\lefteqn{
\forall x^n \in \rs{X}^n,\
\forall y^n\in \rs{Y}^n,\ \forall \theta\in \Theta,} \nonumber\\
 &&\hskip-7mm
  \min_{\ctheta\in \cTheta(x^n)}\bigg\{
  \log\frac{p_{\theta}(y^n|x^n)}{p_{\ctheta}(y^n|x^n)}+\ctpen(\ctheta|x^n)\bigg\}\le
  L_1(\theta|x^n) \label{cvcond2}.\mbox{''}
\end{eqnarray}
      Recall that (\ref{rvcondition}) is
      a sufficient condition for the $\epsilon$-risk validity of $L_1$,
      in fact, it was derived as a sufficient condition for the proposition that
      $L_1(\theta|x^n)$ bounds from above
      \begin{eqnarray}
     &\hskip-1cm E_{\ctheta}[H(\theta,\ctheta,v^n,y^n)]=&
\!\!\!\!\underbrace{E_{\ctheta}\left[
d^n_{\lambda}(p_*,p_{\theta})-
d^n_{\lambda}(p_*,p_{\ctheta})
           \right]
}_{\mbox{(i)}} \nonumber\\
       &&\!\!\!\!+
\underbrace{
E_{\ctheta}\left[
\log
\frac{p_{\theta}(y^n|v^n)}{p_{\ctheta}(y^n|v^n)}
+\ctpen(\ctheta|q_*)\right]}_{\mbox{(ii)}}
\label{eH}
      \end{eqnarray}
    for any $q_*\in \rs{P}^n_x$, $v^n \in \typicalset$, $y^n\in
      \rs{Y}^n$, $\theta\in \Theta$, where $\ctheta$ was randomized on
      $\cTheta(q_*)$ and $(\cTheta(q_*),\ctpen(\ctheta|q_*))$ were defined by
      (\ref{quantizationrd}) and (\ref{ctpen4lasso}), in particular,
$\ctpen(\ctheta|q_*)$ satisfies $\beta$-stronger Kraft's inequality.
  Recall that $H(\theta,\ctheta,x^n,y^n)$ is the inside part of the minimum in (\ref{erv2}).
      Here, we used $v^n$ instead of $x^n$ so as to discriminate from the above
      fixed $x^n$.
To derive the sufficient condition,
we obtained upper bounds on the terms (i) and (ii)
of (\ref{eH})
respectively,
and shown that $L_1(\theta|v^n)$ with $v^n \in A_\epsilon^n$
is not less than the sum of both upper bounds if (\ref{rvcondition}) is satisfied.
A point is that the upper bound on the term (i) we derived is
a non-negative function of $\theta$ (see (\ref{elvp})).
Hence, if $v^n \in A_\epsilon^n$ and (\ref{rvcondition}) hold,
$L_1(\theta|v^n)$ is an upper bound on
the term (ii), which is not less than
\[
  \min_{\ctheta\in \cTheta(q^*)}\bigg\{
  \log\frac{p_{\theta}(y^n|v^n)}{p_{\ctheta}(y^n|v^n)}+\ctpen(\ctheta|q^*)\bigg\}.
\]

Now, assume (\ref{rvcondition}) and
let us take $q_*\in \mathcal{P}^n_x$ given $x^n$, such that
$\Sigma_{jj}$ is equal to $(1/n)\sum_{i=1}^nx^2_{ij}$ for all
$j$.
Then we have $x^n \in A_\epsilon^n$, which implies
\[
L_1(\theta|x^n) \geq
  \min_{\ctheta\in \cTheta(q^*)}\bigg\{
  \log\frac{p_{\theta}(y^n|x^n)}{p_{\ctheta}(y^n|x^n)}+\ctpen(\ctheta|q^*)\bigg\}.
\]
Since
$q^*$ is determined by $x^n$
and $\ctpen(\ctheta|q^*)$ satisfies Kraft's inequality,
the codelength validity condition holds for $L_1$.
  \end{proof}
 \section{Numerical Simulations}\label{simulation}
 We investigate behavior of the regret bound
\eqnum{regretboundlasso}. In the regret bound, we take $\beta=1-\lambda$
with which the regret bound becomes tightest. Furthermore, $\mu_1$ and $\mu_2$ are taken as their
smallest values in (\ref{rvcondition}).
As described before, we cannot obtain the exact bound for KL divergence
which gives the most famous loss function, the mean square error (MSE), in this setting.
This is because the regret bound diverges to the infinity as
$\lambda\rightarrow 1$ unless $n$ is accordingly large enough. That is,
we can obtain only the approximate evaluation of the MSE.
The precision of that approximation depends on the sample size $n$. We
 do not employ the MSE here but another famous loss function, squared
Hellinger distance $d^2_H$ (for a single data). The Hellinger
 distance was defined in (\ref{hellinger}) as $n$ sample version
 (\thatis, $d^2_H=d^{2,1}_H$). 
 We can obtain a regret bound for $d^2_H(p_*,p_{\hat{\theta}})$ by (\ref{regretboundlasso})
 because two times the squared Hellinger distance $2d^2_H$ is bounded by Bhattacharyya
 divergence ($d_{0.5}$) in (\ref{bhattacharyya}) through the
 relationship (\ref{alpharenyi}).  
We set $n=200$, $p=1000$ and $\Sigma=I_p$ to mimic a typical
situation of sparse learning. 
The lasso estimator is calculated by a proximal gradient method
\cite{Beck2009}.
To make the regret bound tight, we take $\tau=0.03$ that is close
to zero compared to the main term (regret).
For this $\tau$,
Fig. \ref{figprob} shows the plot of (\ref{lowerboundP}) against
$\epsilon$.
We should choose the smallest $\epsilon$ as long as the regret bound
holds with large probability.
Our choice is $\epsilon=0.5$ at which the value of (\ref{lowerboundP}) is $0.81$.
We show the results of two cases in Figs. \ref{fig1}-\ref{fig3}. These plots express the value of
$d_{0.5}$, $2d^2_H$ and the regret bound that were
obtained in a hundred of repetitions with different
signal-to-noise ratios (SNR)
$E_{q_*}[(x^T\theta^*)^2]/\sigma^2$ (that is, different $\sigma^2$).
From these figures and other experiments, we observed that
$2d_H^2$ almost always equaled $d_{0.5}$ (they were almost overlapped).
As the SN ratio got larger, then the regret bound became looser, for example, about six times larger than $2d_H^2$ when SNR is $10$. One of
the reasons is that the $\epsilon$-risk validity condition is too strict to bound
the loss function when SNR is high.
Hence, a possible way to improve the risk bound is to restrict the parameter space
$\Theta$ used in $\epsilon$-risk validity to a range of $\hat{\theta}$, which is
expected to be considerably narrower than $\Theta$ due to high SNR.
In contrast, the regret bound is tight when SNR is 0.5 in
Fig. \ref{fig3}.
Finally, we remark that the regret bound dominated the \renyi divergence
over all trials, though the regret bound is probabilistic.
One of the reason is the looseness of the lower bound
$(\ref{lowerboundP})$ of the probability for the regret bound to hold. 
This suggests that $\epsilon$ can be reduced more if we can derive its
tighter bound.
\begin{figure}
\centering
 \includegraphics[width=2.7in]{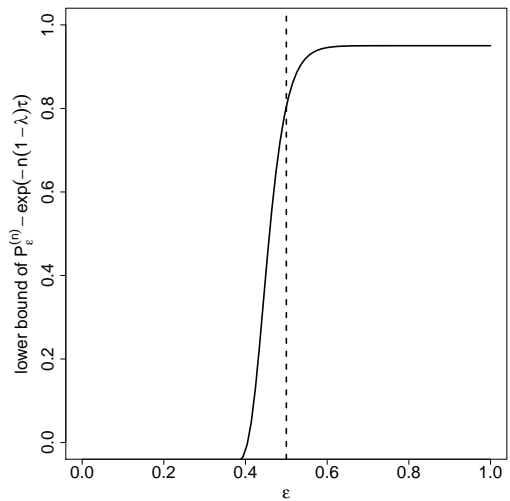}
 \caption{Plot of \eqnum{lowerboundP} against $\epsilon\in (0,1)$ when $n=200,p=1000$ and
 $\tau=0.03$. The dotted vertical line indicates $\epsilon=0.5$.}
\label{figprob}
\end{figure}
\begin{figure}
\centering
 \includegraphics[width=2.7in]{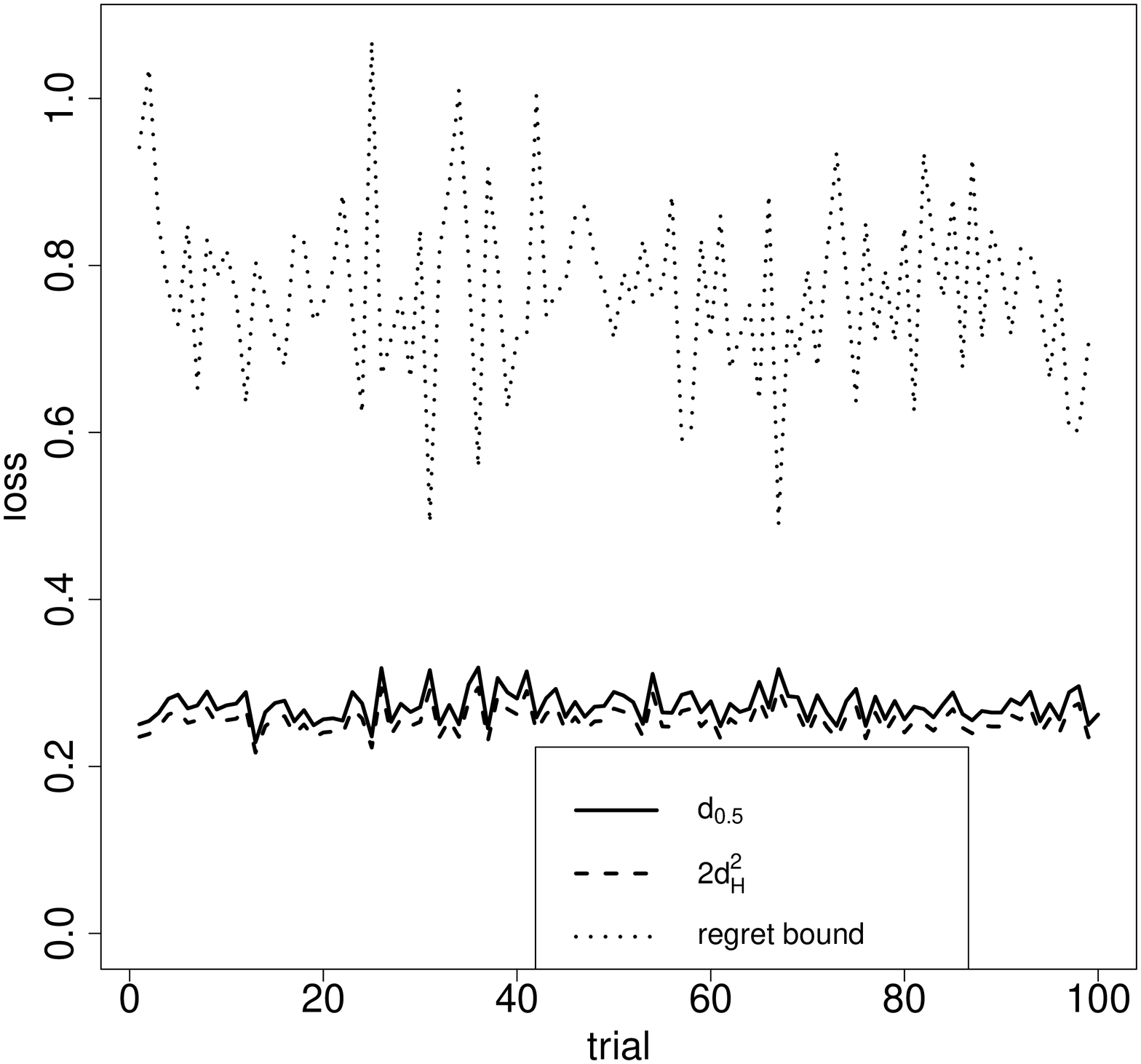}
 \caption{Plot of $d_{0.5}$ (Bhattacharyya div.), $2d_H^2$
 (Hellinger dist.) and the regret bound with
 $\tau=0.03$ in case that SNR=1.5.} 
\label{fig1}
\end{figure}
\begin{figure}
\centering
 \includegraphics[width=2.7in]{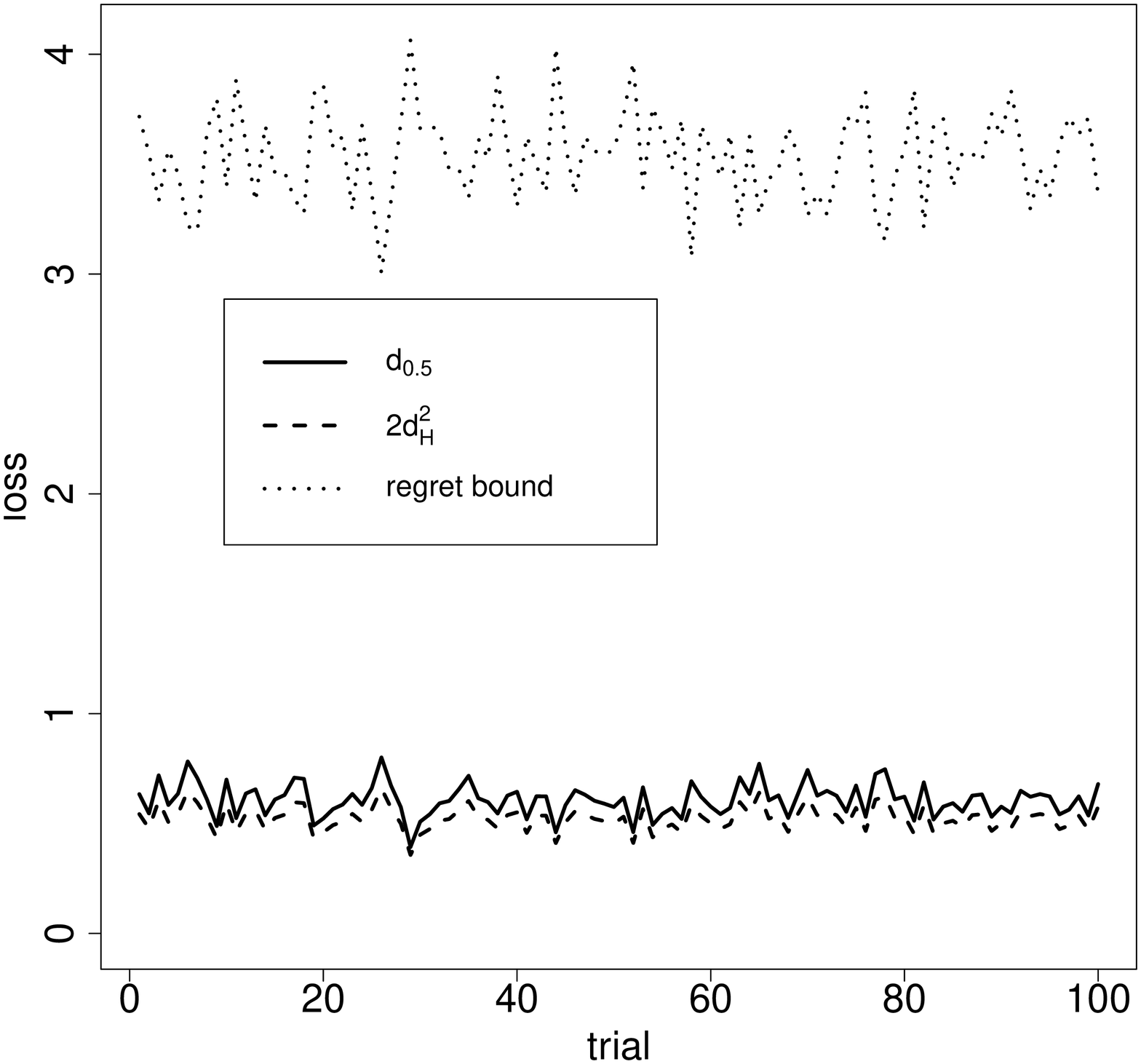}
 \caption{Plot of $d_{0.5}$ (Bhattacharyya div.), $2d_H^2$
 (Hellinger dist.) and the regret bound with
 $\tau=0.03$ in case that SNR=10.} 
\label{fig2}
\end{figure}
\begin{figure}[t]
\centering
 \includegraphics[width=2.7in]{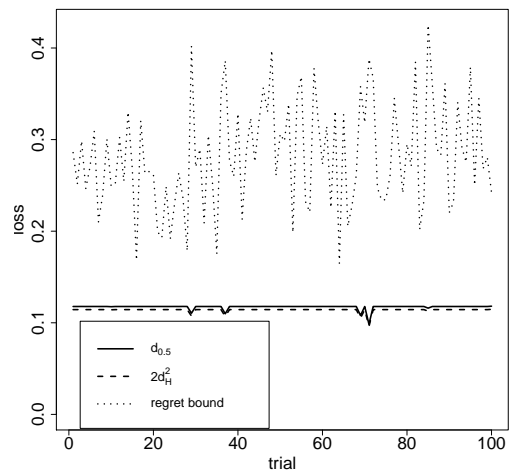}
 \caption{Plot of $d_{0.5}$ (Bhattacharyya div.), $2d_H^2$
 (Hellinger dist.) and the regret bound with
 $\tau=0.03$ in case that SNR=0.5.} 
\label{fig3}
\end{figure}

\section{Conclusion}\label{conclusion}
We proposed a way to extend the original BC theory to supervised
learning by using a typical set. Similarly to the original BC theory,
our extension also gives a mathematical justification of the MDL
principle for supervised learning. As an application, we derived a new
risk and regret bounds of lasso. The derived bounds still retains
various advantages of the original BC theory. In particular, it requires
considerably few assumptions.
Our next challenges are applying our proposal to non-normal cases for lasso and other machine learning methods.

\appendix[Sanov-type Inequality]\label{sanovsection}
The following lemma is a special case of the result in
\cite{csiszar84}. Below, we give a simpler proof. In the lemma, we denote a random variable of one
dimension by $X$ and denote its corresponding one dimensional variable
by $x$. 
  \begin{lemma}\label{exp2bound}
 Let 
\[
 x\sim p_{\theta}(x):=\exp(\theta x-\psi(\theta)),
\]
where $x$ and $\theta$ are of one dimension. Then, 
\begin{eqnarray*}
 \textstyle{\Pr_{\theta}}(X\ge \eta')\le \exp(-\rs{D}(\theta',\theta))&{\rm if
  }\,\,\eta'\ge \eta,\\
 \textstyle{\Pr_{\theta}}(X\le \eta')\le \exp(-\rs{D}(\theta',\theta))& {\rm if
  }\,\,\eta'\le \eta,
\end{eqnarray*}
where $\eta$ is the expectation parameter corresponding to the natural
   parameter $\theta$ and similarly for $\eta'$.
   The symbol $\rs{D}$ denotes the single sample version of the
   KL-divergence defined by (\ref{kldiv}).
  \end{lemma}
   \begin{proof}
In this setting, the KL divergence is calculated as 
\[
 \rs{D}(\theta,\theta')=E_{p_{\theta}}\left[\log\left(\frac{p_{\theta}(X)}{p_{\theta'}(X)}\right)\right]=(\theta-\theta')\eta-\psi(\theta)+\psi(\theta').
\]
  Assume $\eta'-\eta\ge 0$.
 Because of the monotonicity of natural parameter and expectation
 parameter of exponential family, 
\begin{eqnarray*}
 X\ge \eta' &\Leftrightarrow& (\theta'-\theta)X\ge (\theta'-\theta)\eta'\\
  &\Leftrightarrow& \exp\left((\theta'-\theta)X\right)\ge \exp\left((\theta'-\theta)\eta'\right).
\end{eqnarray*}
By Markov's inequality, we have
\begin{eqnarray*}
 \lefteqn{\textstyle{\Pr_{\theta}}\left(\exp\left((\theta'-\theta)X\right)\ge
  \exp\left((\theta'-\theta)\eta'\right)\right)}\\
 &\le&
 \frac{E_{p_{\theta}}\left[\exp\left((\theta'-\theta)X\right)\right]}{\exp\left((\theta'-\theta)\eta'\right)}\\
&=&\int \exp(\theta x-\psi(\theta))\exp((\theta'-\theta)x)dx\cdot \exp(-(\theta'-\theta)\eta')\\
&=&\int \exp(\theta' x-\psi(\theta))dx\cdot \exp(-(\theta'-\theta)\eta')\\
&=&\exp(\psi(\theta'))\exp(-\psi(\theta))\cdot \exp(-(\theta'-\theta)\eta')\\
&=&\exp(-\left((\theta'-\theta)\eta'-\psi(\theta')+\psi(\theta)\right)).
\end{eqnarray*}
The other inequality can also be proved in the same way.
   \end{proof}
 \section*{Inverse matrix formula}
\begin{lemma}\label{inverseformula}
Let $A$ be a non-singular $m\times m$ matrix. If $c$ and $d$ are both $m\times 1$
 vectors and $A+cd$ is non-singular, then
\[
 (A+cd^T)^{-1}=A^{-1}-\frac{A^{-1}cd^TA^{-1}}{1+d^TA^{-1}c}.
\]
\end{lemma}
See, for example, Corollary 1.7.2 in \cite{schott05} for its proof.
\section*{Acknowledgment}
We thank Professor Andrew Barron for fruitful discussion.
The form of \renyi divergence (\ref{andrewsuggestion}) is the
result of simplification suggested by him.  Furthermore, we learned the simple proof of Lemma
\ref{exp2bound} from him. We also thank Mr. Yushin Toyokihara for his support.

\IEEEpeerreviewmaketitle

\ifCLASSOPTIONcaptionsoff
  \newpage
\fi

\bibliographystyle{IEEEtran}

\end{document}